\newif\ifsubmission
\submissionfalse

\ifsubmission
\documentclass{llncs}

\else
\documentclass{article}
\usepackage{fullpage}
\usepackage{palatino}
\fi

\usepackage[T1]{fontenc} 
\usepackage[utf8]{inputenc} 
\usepackage{cmap} 
\usepackage[pdfstartview=FitH,pdfpagemode=None,colorlinks,linkcolor=blue,filecolor=blue,citecolor=Violet,urlcolor=red]{hyperref}
\usepackage{multirow}
\usepackage{breakcites}
\usepackage{float}
\usepackage[section]{placeins}         
\usepackage{mathtools}  
\usepackage[dvipsnames]{xcolor}
\usepackage{amsmath,amsthm,amssymb,algpseudocode,algorithm,cryptocode}
\usepackage{mathrsfs} 
\usepackage[capitalize]{cleveref}
\usepackage{dashbox} 
\usepackage{braket}
\usepackage{physics}
\usepackage{xspace}  
\usepackage[normalem]{ulem}
\usepackage{palatino}
\usepackage{tikz}
\usetikzlibrary{quantikz} 

\newif\ifnotes
\notesfalse
\newcommand{\authnote}[3]{\textcolor{#3}{[{\footnotesize {\bf #1:} { {#2}}}]}}

\newcommand{\andrea}[1]{\ifnotes \authnote{A}{#1}{magenta} \fi}

\newcommand{\gray}[1]{\textcolor{gray}{\mathsf{#1}}}

\ifsubmission

\newtheorem{axiom}[theorem]{Axiom}
\newtheorem{physicsaxiom}[theorem]{Physics Axiom}
\newtheorem{importedtheorem}[theorem]{Imported Theorem}
\newtheorem{informaltheorem}[theorem]{Informal Theorem}
\newtheorem{physicstheorem}[theorem]{Physical Theorem}

\newtheorem{fact}[theorem]{Fact}

\newtheorem{construction}[theorem]{Construction}
\Crefname{importedtheorem}{Imported Theorem}{Imported Theorems}
\Crefname{theorem}{Theorem}{Theorems}
\Crefname{proposition}{Proposition}{Propositions}
\Crefname{claim}{Claim}{Claims}
\Crefname{lemma}{Lemma}{Lemmas}
\Crefname{conjecture}{Conjecture}{Conjectures}
\Crefname{corollary}{Corollary}{Corollaries}
\Crefname{construction}{Construction}{Constructions}
\Crefname{property}{Property}{Properties}

\theoremstyle{definition}

\newtheorem{assumption}[theorem]{Assumption}
\newtheorem{notation}[theorem]{Notation}
\Crefname{definition}{Definition}{Definitions}
\Crefname{assumption}{Assumption}{Assumptions}
\Crefname{notation}{Notation}{Notations}

\theoremstyle{remark}

\newtheorem{comment}[theorem]{Comment}

\Crefname{question}{Question}{Questions}
\Crefname{remark}{Remark}{Remarks}
\Crefname{comment}{Comment}{Comments}
\Crefname{fact}{Fact}{Facts}
\Crefname{step}{Step}{Steps}

\else

\newtheorem{theorem}{Theorem}[section]

\newtheorem{claim}[theorem]{Claim}
\newtheorem{lemma}[theorem]{Lemma}

\newtheorem{corollary}[theorem]{Corollary}

\Crefname{importedtheorem}{Imported Theorem}{Imported Theorems}
\Crefname{theorem}{Theorem}{Theorems}
\Crefname{proposition}{Proposition}{Propositions}
\Crefname{claim}{Claim}{Claims}
\Crefname{lemma}{Lemma}{Lemmas}
\Crefname{conjecture}{Conjecture}{Conjectures}
\Crefname{corollary}{Corollary}{Corollaries}
\Crefname{construction}{Construction}{Constructions}
\Crefname{property}{Property}{Properties}

\theoremstyle{definition}
\newtheorem{definition}[theorem]{Definition}

\Crefname{definition}{Definition}{Definitions}
\Crefname{assumption}{Assumption}{Assumptions}
\Crefname{notation}{Notation}{Notations}

\theoremstyle{remark}

\Crefname{question}{Question}{Questions}
\Crefname{remark}{Remark}{Remarks}
\Crefname{comment}{Comment}{Comments}
\Crefname{fact}{Fact}{Facts}
\Crefname{step}{Step}{Steps}

\fi

\newcommand{\secp}{\lambda}

\def\cA{{\cal A}}

\def\cC{{\cal C}}
\def\cD{{\cal D}}
\def\cE{{\cal E}}
\def\cF{{\cal F}}

\def\cQ{{\cal Q}}
\def\cR{{\cal R}}
\def\cS{{\cal S}}

\def\bbN{{\mathbb N}}

\newcommand{\ba}{\mathbf{a}}

\newcommand{\bm}{\mathbf{m}}
\newcommand{\bx}{\mathbf{x}}
\newcommand{\bv}{\mathbf{v}}
\newcommand{\bw}{\mathbf{w}}
\newcommand{\by}{\mathbf{y}}
\newcommand{\bz}{\mathbf{z}}

\newcommand{\brho}{\boldsymbol{\rho}}
\newcommand{\bsigma}{\boldsymbol{\sigma}}
\newcommand{\bpsi}{\boldsymbol{\psi}}

\def\poly{{\rm poly}}

\def\negl{{\rm negl}}

\newcommand{\com}{\mathsf{com}}

\newcommand{\proref}[1]{Protocol~\protect\ref{#1}}

\newenvironment{boxfig}[2]{\begin{figure}[#1]\fbox{
    \begin{minipage}{\linewidth}
    \vspace{0.2em}\makebox[0.025\linewidth]{}    \begin{minipage}{0.95\linewidth}{{#2 }}
    \end{minipage}\vspace{0.2em}\end{minipage}}}{\end{figure}}

\newcommand{\pprotocol}[4]{
\begin{boxfig}{h!}{
\begin{center}
\textbf{#1}
\end{center}
    #4
\vspace{0.2em} } \caption{\label{#3} #2}
\end{boxfig}
}

\newcommand{\protocol}[4]{
\pprotocol{#1}{#2}{#3}{#4} }

\newcommand{\st}{\mathsf{st}}
\newcommand{\out}{\mathsf{out}}

\newcommand{\zo}{\{0,1\}}
\newcommand{\view}{\mathsf{View}}
\newcommand{\fot}{\mathcal{F}_{\mathsf{OT}}}
\newcommand{\simu}{\mathsf{Sim}}
\newcommand{\commit}{\mathsf{Commit}}

\newcommand{\equivsim}{\mathcal{Q}}
\newcommand{\open}{\mathsf{open}}

\newcommand{\extcommit}{\mathsf{Ext.Commit}}
\newcommand{\extdecommit}{\mathsf{Ext.Decommit}}
\newcommand{\final}{\mathsf{final}}
\renewcommand{\partial}{\mathsf{partial}}

\newcommand{\abort}{\mathsf{abort}}

\newcommand{\eecommit}{\mathsf{EECommit}}
\newcommand{\eedecommit}{\mathsf{EEDecommit}}

\newcommand{\hyb}{\ensuremath{\mathsf{Hybrid}}}

\begin{document}
\title{One-Way Functions Imply Secure Computation\\ in a Quantum World}
\author{James Bartusek\thanks{UC Berkeley. Email: \texttt{bartusek.james@gmail.com}} \and Andrea Coladangelo\thanks{UC Berkeley. Email: \texttt{andrea.coladangelo@gmail.com}} \and Dakshita Khurana\thanks{UIUC. Email: \texttt{dakshita@illinois.edu}} \and Fermi Ma\thanks{Princeton University and NTT Research. Email: \texttt{fermima@alum.mit.edu}}}
\date{}
\maketitle


\begin{abstract}

We prove that quantum-hard one-way functions imply \emph{simulation-secure} quantum oblivious transfer (QOT), which is known to suffice for secure computation of arbitrary quantum functionalities. Furthermore, our construction only makes {\em black-box} use of the quantum-hard one-way function.

Our primary technical contribution is a construction of \emph{extractable and equivocal} quantum bit commitments based on the black-box use of quantum-hard one-way functions in the standard model. 
Instantiating the Crépeau-Kilian (FOCS 1988) framework with these commitments yields simulation-secure QOT. 

\end{abstract}

\ifsubmission
\else
\newpage
{
  \hypersetup{linkcolor=Violet}
  \setcounter{tocdepth}{2}
  \tableofcontents
}
\newpage
\fi


\section{Introduction}
\label{sec:intro}

The complexity of cryptographic primitives is central to the study of cryptography. Much of the work in the field focuses on establishing \emph{reductions} between different primitives, typically building more sophisticated primitives from simpler ones. Reductions imply relative measures of complexity among different functionalities, and over the years have resulted in an expansive hierarchy of assumptions and primitives, as well as separations between them.

\ifsubmission
One-way functions (OWFs) lie at the center of cryptographic complexity: their existence is the minimal assumption necessary for nearly all classical cryptography~\cite{STOC:LubRac86,FOCS:ImpLub89,STOC:ImpLevLub89}. One-way functions are equivalent to so-called ``minicrypt'' primitives like pseudorandom generators, pseudorandom functions and symmetric encryption; 
but provably cannot imply key exchange when used in a black-box way~\cite{C:ImpRud88,C:BarMah09}.
Thus, the existence of key exchange is believed to be a stronger assumption than the existence of one-way functions. Oblivious transfer (OT) 
is believed to be \emph{even stronger}: it implies key exchange, but cannot be obtained from black-box use of a key exchange protocol~\cite{TCC:MahMajPra14}.
\else
One-way functions (OWFs) lie at the center of cryptographic complexity: their existence is the minimal assumption necessary for nearly all classical cryptography~\cite{STOC:LubRac86,FOCS:ImpLub89,STOC:ImpLevLub89}. One-way functions are equivalent to so-called ``minicrypt'' primitives like pseudorandom generators, pseudorandom functions and symmetric encryption; 
but provably cannot imply key exchange when used in a black-box way~\cite{C:ImpRud88,C:BarMah09}.\footnote{In particular,~\cite{C:ImpRud88,C:BarMah09} showed that there cannot exist a key exchange protocol that only has oracle access to the input/output behavior of a one-way function, and makes no additional assumptions. Then~\cite{TCC:Dachman-Soled16} ruled out the possibility of certain types of key exchange protocols that also make use of the {\em code} of a one-way function.  Constructions of key exchange from one-way functions have eluded researchers for decades. This, combined with the aforementioned negative results, is considered to be evidence that key exchange is a qualitatively stronger primitive than one-way functions in the classical regime. In fact, Impagliazzo~\cite{Impagliazzo} stipulates that we live in one of five possible worlds, of which {\em Minicrypt} is one where classical one-way functions exist but classical public-key cryptography does not.}
Thus, the existence of key exchange is believed to be a stronger assumption than the existence of one-way functions. Oblivious transfer (OT) 
is believed to be \emph{even stronger}: it implies key exchange, but cannot be obtained from black-box use of a key exchange protocol~\cite{TCC:MahMajPra14}.\fi

The importance of OT stems from the fact that it can be used to achieve secure computation, which is a central cryptographic primitive with widespread applications.
In a nutshell, secure computation allows mutually distrusting participants to compute any public function over their joint private inputs while revealing no private information beyond the output of the computation.



\paragraph{The Quantum Landscape.}
The landscape of cryptographic possibilities changes significantly when participants have quantum computation and communication capabilities. 
For one, {\em unconditionally} secure key distribution --- commonly known as \emph{quantum key distribution} (QKD) --- becomes possible \cite{BenBra84}. 
Moreover, 
\emph{quantum} oblivious transfer (QOT) is known to be achievable from special types of commitments, as we discuss next.

Crépeau and Kilian~\cite{FOCS:CreKil88} first proposed a protocol for QOT using quantum bit commitments.
The central idea in these QKD and QOT protocols is the use of (what are now known as) ``BB84 states''. These are single qubit states encoding either $0$ or $1$ in either the computational or Hadamard basis. Crucially, measuring (or essentially attempting to copy the encoded bit) in the wrong basis completely destroys information about the encoded bit. Then~\cite{C:BBCS91} presented a transmission-error resistant version of the~\cite{FOCS:CreKil88} protocol. These protocols did not come with a proof of security, but subsequently 
Mayers and Salvail~\cite{MayersSalvail94} proved that the~\cite{FOCS:CreKil88} protocol is secure against a restricted class of attackers that only perform single-qubit measurements. This was later improved by Yao~\cite{STOC:Yao95}, who extended the~\cite{MayersSalvail94} result to handle general quantum adversaries.

By an unfortunate historical accident, the aforementioned security proofs claimed the \cite{FOCS:CreKil88} QOT could be \emph{information-theoretically secure}, since at the time it was believed that information-theoretic quantum bit commitment was possible~\cite{FOCS:BCJL93}. Several years later, Mayers~\cite{Mayers97} and Lo and Chau~\cite{LoChau97} independently proved the impossibility of information-theoretic quantum bit commitment, and as a consequence, the precise security of~\cite{FOCS:CreKil88} QOT was once again unclear. This state of affairs remained largely unchanged until 2009, when Damgard, Fehr, Lunemann, Salvail, and Schaffner \cite{C:DFLSS09} proved that bit commitment schemes satisfying certain additional properties, namely \emph{extraction and equivocation}, suffice to instantiate~\cite{FOCS:CreKil88} QOT.~\cite{C:DFLSS09} called their commitments {\em dual-mode} commitments, and provided a construction based on the quantum hardness of the learning with errors (QLWE) assumption. We remark that assumptions about the hardness of specific problems like QLWE are qualitatively even worse than general assumptions like QOWFs and QOT. Thus, the following basic question remains:
\begin{center}
        {\em Do quantum-hard one-way functions suffice for quantum oblivious transfer?}
\end{center}

\paragraph{Quantum OT: The Basis of Secure Quantum Computation.}
There is a natural extension of secure computation to the quantum world, where Alice and Bob wish to compute a \emph{quantum} circuit on (possibly entangled) \emph{quantum} input states.
This setting,
usually referred to as secure \emph{quantum} computation, has been previously studied and in fact has a strong tradition in the quantum cryptography literature. 


\cite{STOC:CreGotSmi02,FOCS:BCGHS06} constructed unconditional maliciously-secure \emph{multi-party} quantum computation with honest majority. The setting where half (or more) of the players are malicious requires computational assumptions due to the impossibility of unconditionally secure quantum bit commitment~\cite{Mayers97,LoChau97}. 

In this computational setting,~\cite{C:DupNieSal10,DNS12} showed the feasibility of two-party quantum computation (2PQC) assuming post-quantum OT.
More recently,~\cite{EC:DGJMS20} constructed maliciously-secure general multi-party quantum computation (MPQC) secure against a \emph{dishonest} majority from any maliciously-secure post-quantum multi-party computation (MPC) protocol for classical functionalities, which can itself be obtained from post-quantum OT~\cite{ABGKM20}. 

Nevertheless, the following natural question has remained unanswered:
\begin{center}
{\em Can secure (quantum) computation be obtained from quantum-hard one-way functions?}
\end{center}
%

\subsection{Our Results}
\label{subsec: our results}
Our main result is the following:
\begin{center}
{\em 
Quantum oblivious transfer can be based on the 
assumption that quantum-hard one-way functions exist. }
\end{center}

In fact, we prove a stronger result: 
we show that quantum oblivious transfer can be based on the {\em black-box use of any statistically binding, quantum computationally hiding commitment}. Such commitments can be based on the black-box use of quantum-hard one-way functions.
This in turn implies secure two-party computation of classical functionalities, in the presence of quantum computation and communication capabilities, from (black-box use of) quantum-hard one-way functions~\cite{STOC:Kilian88}. The latter can then be used to obtain secure two-party \emph{quantum} computation, by relying on the work of~\cite{DNS12}.
Quantum OT can also be used to obtain {\em multi-party} secure computation of all classical functionalities, in the presence of quantum computation and communication capabilities, and additionally assuming the existence of authenticated channels. This follows from the techniques in~\cite{STOC:Kilian88,C:CreVanTap95,C:DamIsh05,C:IshPraSah08} which obtain classical MPC based on black-box use of any OT protocol.
By relying on~\cite{EC:DGJMS20}, this also implies multi-party secure {\em quantum} computation. 

In summary, our main result implies that: (1) 2PQC can be obtained from (black-box use of) quantum-hard OWFs and  (2) assuming the existence of authenticated channels, MPQC can be obtained from (black-box use of) quantum-hard OWFs.

This gives a potential separation between the complexity of cryptographic primitives in the classical and quantum worlds. In the former, (two-party) secure computation provably cannot be based on black-box use of quantum-hard one-way functions. It is only known from special types of enhanced public-key encryption schemes or from the hardness of specific problems, both of which are believed to be much stronger assumptions than one-way functions. But in the quantum world, prior to our work, (two-party) secure computation was only known from the special commitments required in the protocol of \cite{C:DFLSS09}, which can be based on QLWE following~\cite{C:DFLSS09}, or post-quantum OT  (implicit in~\cite{C:HalSmiSon11,STOC:BS20,ABGKM20}) --- but were not known to be achievable from quantum-hard one-way functions.


\paragraph{On the Significance of the Black-Box use of Cryptography in the Quantum Setting.}
Making black-box use of a cryptographic primitive refers to only having oracle access to its input/output behavior, without having the ability to examine the actual code (i.e., representation as a sequence of symbols) of the primitive. For instance, proving in zero-knowledge that a committed value satisfies some given predicate often requires explicit knowledge of the commitment algorithm; classifying the resulting proof as making ``non-black-box'' use of the one-way function. 
In the literature, constructions that make black-box use of cryptographic primitives are often  preferred over those that make non-black-box use of cryptography. Besides their conceptual simplicity and elegance, black-box constructions are also of practical interest since they avoid expensive NP reductions involving circuits of primitives.
Perhaps most importantly, in the case of black-box constructions, one can instantiate the underlying primitive with an arbitrary implementation, including physical implementations via secure hardware or those involving quantum communication.

In many quantum protocols, which involve quantum states being transferred between two or more players, black-box constructions are not only significantly preferable but often become {\em a necessity}. Let us illustrate this with an example. The GMW protocol~\cite{C:GolMicWig86} first showed that secure multi-party computation can be based on any oblivious transfer protocol; however the protocol involved zero-knowledge proofs involving the description of the (classical) oblivious transfer protocol. Due to the GMW~\cite{C:GolMicWig86} protocol being non-black-box in the underlying OT, our OT protocols cannot be used with GMW to obtain multi-party computation of classical functionalities. We instead need to rely on compilers like~\cite{STOC:Kilian88,C:CreVanTap95,C:DamIsh05,C:IshPraSah08} that only make {\em black-box use} of the underlying OT protocol. As discussed above, the black-box nature of these compilers makes them applicable irrespective of whether they are instantiated with classical or QOT.

In a similar vein, we believe that our {\em black-box} use of any statistically binding, quantum computationally hiding commitment in our QOT protocol is of particular significance. 
For instance, one can substitute our statistically binding, quantum computationally hiding commitment with an unconditionally secure one in the quantum random oracle model~\cite{C:Unruh14}, resulting in {\em unconditional quantum OT in the quantum random oracle model}.
Moreover if in the future, new constructions of statistically binding, quantum computationally hiding commitments involving quantum communication are discovered based on assumptions weaker than quantum-hard one-way functions, it would be possible to plug those into our protocol compilers to obtain QOT. 
These applications would not have been possible had we required non-black-box use of the underlying commitment.

\paragraph{Primary Tool: Stand-alone Extractable and Equivocal Commitments.}
As discussed earlier, \cite{C:DFLSS09} show that simulation-secure QOT can be obtained from commitments satisfying certain properties, namely \emph{extraction} and \emph{equivocation}. 
\begin{itemize}
    \item At a high level, extraction requires that there exist an efficient quantum ``extractor'' that is able to extract a committed message from any quantum committer.  
    \item Equivocality requires that there exist an efficient quantum ``equivocator'' capable of simulating an interaction with any quantum receiver such that it can later open the commitment to any message of its choice.
\end{itemize}

These two properties are crucial for proving simulation security of the \cite{FOCS:CreKil88} OT protocol: extraction implies receiver security and equivocality implies sender security\footnote{It is important to note that extraction and equivocation are only made possible in an ideal world where a simulator has access to the adversary's state. Participants in the real protocol cannot access each others' state, which prevents them from extracting or equivocating.}. Our key technical contribution is the following:
\begin{center}
    {\em Extractable and equivocal commitments can be based on the black-box use of\\ quantum-hard one-way functions.}
\end{center}
We obtain this result via the following transformations, each of which only makes \emph{black-box} use of the underlying primitives.
\begin{itemize}
    \item {\em Step 1: Quantum Equivocal Commitments from Quantum-Hard One-Way Functions.} We describe a generic unconditional compiler to turn any commitment into an equivocal commitment in the plain model. 
    By applying our compiler to Naor's statistically binding commitment~\cite{JC:Naor91} --- which can be based on quantum-hard one-way functions --- we obtain a statistically binding, equivocal commitment.

    \item {\em Step 2: Quantum Extractable Commitments from Quantum Equivocal Commitments.}
    We show that  the~\cite{FOCS:CreKil88,C:DFLSS09,C:BouFeh10} framework can be used to obtain an extractable commitment that leverages quantum communication, and can be based on the existence of any quantum equivocal commitment.
    This combined with the previous step implies the existence of quantum extractable commitments based on the existence of quantum-hard one-way functions. 
    
    This is in contrast to existing approaches (eg.,~\cite{C:HalSmiSon11}) that require classical communication but rely on qualitatively stronger assumptions like classical OT with post-quantum security.
    \item {\em Step 3: From Extractable Commitments to Extractable and Equivocal Commitments.} We apply the black-box equivocality compiler from the first step to the quantum extractable commitment obtained above, to produce an extractable and equivocal commitment. 
    
    We point out that it is generally straightforward to make a classical commitment equivocal using zero-knowledge proofs, but this approach does not apply to quantum commitment protocols. We therefore devise our own equivocality compiler capable of handling quantum commitments and use it in both Step 1 and Step 3.

\end{itemize}

Plugging our quantum extractable and equivocal commitments into the~\cite{FOCS:CreKil88} framework yields a final QOT protocol with an interaction pattern that readers familiar with~\cite{BenBra84,FOCS:CreKil88} may find interesting: the sender sends the receiver several BB84 states, after which the receiver proves to the sender that it has honestly measured the sender's BB84 states by  \emph{generating more BB84 states of their own and asking the sender to prove that they have measured the receiver's BB84 states.}
An intriguing open question is whether obtaining QOT from one-way functions \emph{requires} this type of two-way quantum communication or, alternatively, quantum memory.\footnote{Naive approaches to removing one direction of quantum communication appear to require the honest parties to be entangled and subsequently perform quantum teleportation.}

\subsection{Related Work}
\label{subsec:related work}

For some readers, it may appear that the central claim of this work --- that quantum-hard one-way functions suffice for oblivious transfer --- has already been established~\cite{C:DFLSS09,C:BouFeh10}. Indeed, prior work~\cite{C:DFLSS09,C:BouFeh10} showed that statistically binding and computational hiding commitments (which are weaker than extractable and equivocal commitments), known to exist from one-way functions, can be plugged into the~\cite{FOCS:CreKil88} template to achieve an oblivious transfer protocol satisfying \emph{indistinguishability-based} security. 

However, the indistinguishability-based security definition for oblivious transfer is \emph{not} standard in the cryptographic literature. When cryptographers refer to ``oblivious transfer'', they almost always mean the standard \emph{simulation-based} security notion. Indeed, the fundamental importance of oblivious transfer in modern cryptography is due to the fact that it is necessary and sufficient for secure computation, but this is only true for the simulation-based notion.

\ifsubmission
\else
Therefore, the goal of this work is to build simulation-secure oblivious transfer, and thereby establish that quantum-hard one-way functions suffice for secure quantum computation. We elaborate further on the details of simulation security in~\cref{sec:techniques}.

We also note that in the setting of Universal Composability (UC), Unruh~\cite{C:Unruh07} obtained quantum oblivious transfer from quantum UC-secure commitments in the common reference string (CRS) model. By contrast, in this work, we focus on the standalone setting in the plain/standard model, and we do not assume the existence of trusted setup or a common reference string. In this setting, we obtain quantum oblivious transfer from quantum-hard one-way functions. 
\fi

\subsection{Concurrent and Independent Work}
In a concurrent and independent work, Grilo, Lin, Song, and Vaikuntanathan~\cite{GLSV} also construct simulation-secure quantum oblivious transfer from quantum-hard one way functions via the intermediate primitive of extractable and equivocal commitments. However, the two works take entirely different approaches to constructing these commitments. We briefly summarize these strategies below.\\

\noindent\underline{This work:}
\begin{enumerate}
\item Construct equivocal commitments from statistically binding commitments via a new ``equivocality compiler'' based on Watrous~\cite{STOC:Watrous06} rewinding.
\item Construct extractable commitments from equivocal commitments via a new “extractability compiler” based on the~\cite{FOCS:CreKil88} template.
\item Construct extractable and equivocal commitments from extractable commitments via the same compiler from Step 1.
\end{enumerate}

\noindent\underline{\cite{GLSV}}:
\begin{enumerate}
\item Construct selective opening secure commitments with inefficient simulation against malicious committers from statistically binding commitments and zero-knowledge proofs.
\item Construct QOT with inefficient simulation against malicious receivers from selective opening secure commitments with inefficient simulation against malicious committers, following the~\cite{FOCS:CreKil88} QOT template.\footnote{\cite{GLSV} point out that the conclusions of Steps 1 and 2 together had also been established in prior works of~\cite{C:Unruh07,C:DFLSS09,C:BouFeh10}.}
\item Construct parallel QOT with inefficient simulation against malicious receivers from (stand-alone) QOT with inefficient simulation against malicious receivers via a new lemma for parallel repetition of protocols.
\item Construct verifiable conditional disclosure of secrets, a new primitive introduced in~\cite{GLSV}, from parallel QOT with inefficient simulation against malicious receivers, statistically binding commitments, Yao’s garbled circuits, and zero-knowledge proofs.
\item Construct extractable commitments from verifiable conditional disclosure of secrets, statistically binding commitments, and zero-knowledge proofs.
\item Construct extractable and equivocal commitments from extractable commitments and zero-knowledge proofs.
\end{enumerate}

We believe that our result is easier to understand and conceptually simpler, as we do not need to define additional primitives beyond extractable and/or equivocal commitments. Aside from differences in approach, there are several other places where the results differ:
\begin{itemize}
\item \textbf{This Work: Black-Box Use of One-Way Functions.} A significant advantage of our work over~\cite{GLSV} is that we construct quantum OT from {\em black-box use} of statistically binding commitments or one-way functions. The OT in~\cite{GLSV} makes non-black-box use of the underlying one-way function.
As discussed above, making black-box use of underlying cryptographic primitives is particularly useful in the quantum setting. Due to the extensive use of zero-knowledge proofs and garbled circuits in~\cite{GLSV}, it appears difficult to modify their approach to be black-box in the underlying one-way function.
\item \textbf{This Work: One-Sided Statistical Security.} Additionally, our oblivious transfer protocol offers {\em one-sided statistical security}. As written, our quantum OT protocol satisfies statistical security against malicious senders (and computational security against malicious receivers). Moreover, this OT can be reversed following the techniques in eg.,~\cite{EC:WolWul06} to obtain a quantum OT protocol that satisfies statistical security against malicious receivers (and computational security against malicious senders). 
On the other hand, the quantum OT protocol in~\cite{GLSV} appears to achieve computational security against both malicious senders and malicious receivers.
\item \textbf{\cite{GLSV}: Verifiable Conditional Disclosure of Secrets.} Towards achieving their main result, \cite{GLSV} introduce and construct verifiable conditional disclosure of secrets (vCDS). This primitive may be of independent interest.
\item \textbf{\cite{GLSV}: Constant Rounds in the CRS Model.} While both works construct $\poly(\secp)$-round protocols in the plain model, \cite{GLSV} additionally construct a constant round OT protocol in the CRS model based on (non-black-box use of) quantum-hard one-way functions. 

\ifsubmission
\emph{In an earlier version of this work, we did not consider the CRS model. After both works were posted to the Cryptology ePrint Archive, 
we realized that our techniques could be straightforwardly adapted to achieve constant round complexity in the CRS model, while still making black-box use of one-way functions. However, unlike \cite{GLSV}, our CRS is non-reusable. For the interested reader, we sketch how this can be achieved in the full version. 
}
\else
\emph{In an earlier version of this work, we did not consider the CRS model. After both works were posted to the Cryptology ePrint Archive, 
we realized that our techniques could be adapted straightforwardly to achieve constant round complexity in the CRS model, while still making black-box use of one-way functions. We elaborate further on this below.}
\fi
\end{itemize}

\ifsubmission
\else
\paragraph{Constant Rounds in the CRS Model via Our Techniques.} We now briefly sketch how the protocols in this work can be modified to achieve constant rounds in the CRS model. The following two paragraphs are the only component of this work that was not developed concurrently and independently of \cite{GLSV}.

In a nutshell, the reason why our protocols require super-constant rounds in the plain model is because they perform sequential repetitions of specific $\Sigma$-protocols to enable rewinding.
In the presence of a CRS, it is possible to repeat our protocols in parallel and simulate by means of a trapdoor embedded in the CRS. This avoids the need for rewinding (and therefore sequential repetition), and would bring down the number of rounds from $O(\lambda)$ to $O(1)$ in the CRS model.

In fact, constant round QOT can also be based on the {\em black-box use} of quantum-hard one-way functions in the presence of a CRS by slightly modifying our template, as follows.
Recall that the first step towards getting a constant round protocol is to replace (all of) our sequential repetitions with parallel repetitions.
The challenge for the parallel repetitions can be determined as the outcome of a (one-sided simulatable) coin toss. This coin toss can be carefully designed so that a simulator that programs the CRS can program the output of the coin toss, which would allow the simulator to program the challenge message for the parallel repetition. 
In the CRS model, such a coin toss can be obtained by having one player commit to a random set of coins via Naor's commitment~\cite{JC:Naor91} (where the first message of the commitment is the CRS). Next, the second player sends their random coins, in the clear, to the first player. Finally, the first player decommits their initial commitment string, and the players set the output of the coin toss to be the XOR of the random coins chosen by each player. A simulator can program the output of the coin toss by programming the CRS to equivocate Naor's commitment. This programming would allow one to simulate the parallel repeated versions of our protocols. 
We note that this protocol has the drawback of the CRS not being reusable, unlike the CRS-based protocol in~\cite{GLSV}.
We stress that our focus is on the plain model, and we leave formal explorations/optimizations of this and related  approaches to building black-box OT in the CRS model to future work.

\fi

\section{Technical Overview}
\label{sec:techniques}

This work establishes that (1) black-box use of post-quantum one-way functions suffices for post-quantum \emph{extractable and equivocal} commitment schemes and moreover, that (2)~\cite{FOCS:CreKil88} quantum oblivious transfer instantiated with such commitments is a standard \emph{simulation-secure} oblivious transfer. Crucially, the standard notion of simulation-secure (quantum) oblivious transfer that we achieve is sequentially composable and suffices to achieve general-purpose secure quantum computation.
Before explaining our technical approach, we provide a complete review of the original~\cite{FOCS:CreKil88} protocol.

\subsection{Recap: Quantum Oblivious Transfer from Commitments}
\label{subsec:bbcs}

In quantum oblivious transfer (QOT), a quantum sender holding two classical messages $m_0,m_1$ engages in an interactive protocol over a quantum channel with a quantum receiver holding a classical choice bit $b$. Correctness requires the receiver to learn $m_b$ by the end of the protocol.
Informally, security demands that a malicious receiver only learn information about one of $m_0,m_1$, and that a malicious sender learn nothing about $b$.
Somewhat more formally, as discussed earlier, our focus is on the standard simulation-based notion of security. This stipulates the existence of an efficient quantum simulator that generates the view of an adversary (sender/receiver) when given access to an ideal OT functionality. In particular, when simulating the view of a malicious sender, this simulator must extract the sender's inputs $(m_0, m_1)$ without knowledge of the receiver's input $b$. And when simulating the view of a malicious receiver, the simulator must extract the receiver's input $b$, and then simulate the receiver's view given just $m_b$.

We recall the construction of quantum oblivious transfer due to~\cite{FOCS:CreKil88} (henceforth CK88), which combines the information theoretic quantum key distribution protocol of~\cite{BenBra84} (henceforth BB84) with cryptographic bit commitments.

\paragraph{CK88 First Message.} The first message of the CK88 protocol exactly follows the beginning of the BB84 protocol. For classical bits $y,z$, let $\ket{y}_z$ denote $\ket{y}$ if $z = 0$, and $(\ket{0} + (-1)^y\ket{1})/\sqrt{2}$ if $z = 1$, i.e. the choice of $z$ specifies whether to interpret $y$ as a computational or Hadamard basis vector. Let $\lambda$ denote the security parameter. The sender samples two random $2\lambda$-bit strings $x$ and $\theta$, and constructs ``BB84 states'' $\ket{x_i}_{\theta_i}$ for $i \in [2\lambda]$. The sender forwards these $2\lambda$ BB84 states $(\ket{x_i}_{\theta_i})_{i \in [2\lambda]}$ to the receiver. Next, the receiver samples a $2\lambda$-bit string $\hat{\theta}$, measures each $\ket{x_i}_{\theta_i}$ in the basis specified by $\hat{\theta}_i$, and obtains a $2\lambda$-bit measurement result string $\hat{x}$.

\paragraph{CK88 Measurement-Check Subprotocol.} At this point, the CK88 and BB84 protocols diverge. Since the BB84 protocol is an interaction between two \emph{honest} parties, it assumes the parties comply with the protocol instructions. However, in the CK88 protocol, a malicious receiver who does not measure these BB84 states will be able to compromise sender privacy later in the protocol. Therefore, the next phase of CK88 is a measurement-check subprotocol designed to catch a malicious receiver who skips the specified measurements. This subprotocol requires the use of a quantum-secure classical commitment scheme; for the purposes of this recap, one should imagine a commitment with idealized hiding and binding properties. The subprotocol proceeds as follows:
\begin{itemize}
\item For each $i \in [2\lambda]$, the receiver commits to $(\hat{\theta}_i,\hat{x}_i)$.
\item Next, the sender picks a random set $T$ of $\lambda$ indices from $[2\lambda]$, and challenges the receiver to open the corresponding commitments.
\item The receiver sends $(\hat{\theta}_i,\hat{x}_i)$ along with the corresponding opening for each $i \in T$. 
\item The sender verifies each commitment opening, and furthermore checks that $\hat{x}_i = x_i$ for each $i \in T$ where $\hat{\theta}_i = \theta_i$. If any of these checks fail, the sender aborts.
\end{itemize}

The rough intuition for the subprotocol is simple: from the receiver's point of view, the BB84 states are maximally mixed and therefore completely hide $x_i$ and $\theta_i$. For any index $i$ that the receiver does not measure, it must guess $\hat{x}_i$. From the receiver's perspective, the sender checks $\hat{x}_i$ against $x_i$ if two $1/2$-probability events occur: (1) $i$ is included in $T$, and (2) $\hat{\theta}_i = \theta_i$. This means a malicious receiver who skips a significant number of measurements will be caught with overwhelming probability.

\paragraph{CK88 Privacy Amplification.} If all the subprotocol checks pass, the sender continues to the final stage of the CK88 protocol. For convenience, relabel the $\lambda$ indices in $[2\lambda] \setminus T$ from $1$ to $\lambda$; all indices corresponding to opened commitments are discarded for the remainder of the protocol.

For each $i \in [\lambda]$, the sender reveals the correct measurement basis $\theta_i$. The receiver then constructs the index set $I_b$ --- where $b$ is its choice bit for the oblivious transfer --- as the set of all $i \in [\lambda]$ where $\theta_i = \hat{\theta}_i$. It sets $I_{1-b}$ to be the remaining indices, and sends $(I_0, I_1)$ to the sender. Note that by the hiding property of the commitments, the sender should not be able to deduce $b$ from $(I_0,I_1)$; furthermore, $I_0$ and $I_1$ will both be close to size $\lambda/2$, since for each $i \in [\lambda]$, the receiver committed to $\hat{\theta}_i$ before obtaining $\theta_i$.

On receiving $I_0,I_1$, the sender sets $x_0 \coloneqq (x_i)_{i \in I_0}$ and $x_1 \coloneqq (x_i)_{i \in I_1}$. The intuition is that if a receiver honestly constructs $(I_0, I_1)$, it will only have information about $x_b$ corresponding to its choice bit $b$. However, it turns out that even if the receiver maliciously constructs $(I_0,I_1)$, at least one of $x_0$ and $x_1$ will have high min-entropy from its point of view. Thus, by standard privacy amplification techniques, the sender can complete the oblivious transfer as follows. It samples two universal hash functions $h_0$ and $h_1$, both with $\ell$-bit outputs, and uses $h_0(x_0)$ to mask the $\ell$-bit message $m_0$, and uses $h_1(x_1)$ to mask $m_1$. That is, the sender sends $(h_0,h_1,h_0(x_0) \oplus m_0, h_1(x_1) \oplus m_1)$ to the receiver, who can then use $x_b$ to recover $m_b$. Since $x_{1-b}$ will have high entropy, the leftover hash lemma implies that $h_{1-b}(x_{1-b})$ is statistically close to uniform, which hides $m_{1-b}$ from the receiver.

\paragraph{Simulation-Based Security.}
Turning this intuition into a proof of simulation-based security of the resulting QOT requires some additional insights~\cite{C:DFLSS09}, and requires the commitments used in the measurement-check subprotocol to satisfy two additional properties: extractability and equivocality. In what follows, we briefly summarize why these properties help achieve simulation-based security.

To argue that the resulting QOT protocol satisfies security against a malicious sender, one must demonstrate the existence of a simulator that simulates the sender's view by generating messages on behalf of an honest receiver, and extracts both QOT inputs of the sender\footnote{We refer the reader to Section~\ref{sec:otdef} for a formal definition of simulation-based QOT.}. Now, the measurement-check subprotocol described above is designed to ensure that at least one of the sender's inputs is hidden from a receiver. To nevertheless enable the simulator to extract both sender inputs, the idea in~\cite{C:DFLSS09} is to substitute the commitments used in the measurement-check subprotocol with {\em equivocal commitments} that allow the simulator to later open these commitments to any value of its choice. This enables the simulator to defer any measurements until after it obtains the set $T$ from the sender, and then selectively measure {\em only} the states that correspond to indices in $T$. All other states are left untouched until the sender reveals its measurement bases in the final stage of the CK88 protocol.
Upon obtaining the sender's ``correct'' measurement bases, the simulator measures all the remaining states in the correct bases, allowing it to learn both the inputs of the sender.

To demonstrate that the resulting QOT protocol satisfies security against a malicious receiver, one must demonstrate the existence of a simulator that simulates the receiver's view by generating messages on behalf of an honest sender, and extracts the receiver's choice bit. 
Again by design, the measurement-check subprotocol ensures that the receiver's choice bit hidden is hidden from the sender. To nevertheless enable the simulator to extract this choice bit,~\cite{C:DFLSS09} modify the commitments in the measurement-check subprotocol so that the simulator is able to extract all of the $\{(\widehat{\theta}_i, \widehat{x}_i)\}_{i \in [2\lambda]}$ from the receiver's commitments. This enables the simulator to compute which one of the sets $I_0, I_1$ contain more indices $i$ for which $\theta_i = \widehat{\theta}_i$; clearly the set with more indices corresponds to the receiver's choice bit.
In summary, the key tool that enables simulation against a malicious receiver is an {\em extractable} commitment, that forces the receiver to use commitments for which the simulator can extract the committed value, without ever running the opening phase.

To conclude, following~\cite{C:DFLSS09} the CK88 protocol can be shown to satisfy simulation-based security as long as the commitments used in the measurement-check subprotocol satisfy both the {\em extractability} and {\em equivocality} properties that were informally described above.


With this in mind, we now describe our primary technical contribution: a construction of the required extractable and equivocal commitments based on black-box use of quantum-hard one-way functions. 

\subsection{Our Construction: A High-Level Overview}
\label{subsec:overview}

The rest of this technical overview describes our \emph{black-box} construction of simultaneously \emph{extractable and equivocal} quantum bit commitments from any quantum-hard one-way function. 

The ingredients for our construction are the following:
\begin{itemize}
\item A general-purpose ``equivocality compiler'' that turns any bit commitment scheme --- classical or quantum --- into an \emph{equivocal} quantum commitment scheme. Moreover, if the original commitment scheme is \emph{extractable}, this compiler outputs an \emph{extractable and equivocal} commitment scheme. 
\item A general-purpose ``extractability compiler'' that turns any \emph{equivocal} bit commitment scheme --- classical or quantum --- into an \emph{extractable but not equivocal} commitment scheme.
\end{itemize}
Both of these compilers require no additional computational assumptions beyond those of the original commitment schemes. Given these compilers, we build extractable and equivocal commitments via the following steps:
\begin{itemize}
\item \textbf{Instantiation:} Begin with Naor's statistically-binding, computationally hiding commitments~\cite{JC:Naor91}. 
Naor's construction makes black-box use of one-way functions and achieves post-quantum computational hiding assuming post-quantum security of the one-way function.\footnote{In slightly more detail, Naor's commitment scheme makes black-box use of any pseudo-random generator (PRG). It is straightforward to verify that if the PRG is post-quantum secure, the commitment satisfies computational hiding against quantum attackers. A black-box construction of pseudo-random generators from one-way functions is due to~\cite{HILL99}; Aaronson~\cite{Aaronson09} and Zhandry~\cite{FOCS:Zhandry12} observed that~\cite{HILL99} applies to non-uniform quantum attackers with \emph{classical} advice. This can be extended to handle non-uniform {quantum} advice by giving the one-way function attacker constructed in the~\cite{HILL99} reduction many copies of the PRG attacker's non-uniform quantum advice (which only requires some polynomial upper bound on the number of times the reduction invokes the PRG attacker).}
\item \textbf{Step 1:} Plug Naor's commitments into our equivocality compiler to obtain an \emph{equivocal} quantum bit commitment scheme.
\item \textbf{Step 2:} Feed the resulting equivocal quantum bit commitments into our extractability compiler to obtain an \emph{extractable but not equivocal} quantum bit commitment.
\item \textbf{Step 3:} Run the equivocality compiler \emph{a second time}, but now starting with the extractable commitments produced by the previous step. This gives the desired \emph{extractable and equivocal} quantum bit commitments.
\end{itemize}

\ifsubmission
\else
We briefly remark that if one did not care about obtaining a construction that only makes \emph{black-box} use of one-way functions, it would only be necessary to invoke the equivocality compiler once. In slightly more detail, we could design an alternative extractability compiler that works with any (not necessarily equivocal) bit commitment scheme by additionally relying on post-quantum zero-knowledge proofs~\cite{STOC:Watrous06}; this approach would still give a construction of quantum oblivious transfer from post-quantum one-way functions, but the zero-knowledge proofs would be used to prove statements involving the \emph{description} of some one-way function, which is a non-black-box use of the primitive. 
We believe that constructing QOT from quantum-hard one-way functions is interesting regardless of whether or not the construction is black-box. 
However, as discussed in Section \ref{subsec: our results}, black-box constructions are significantly preferred over non-black-box ones, especially in the quantum setting. Therefore our focus in the rest of this overview will be on presenting a black-box construction.
\fi

\ifsubmission
\else
\paragraph{Organization.} We describe our equivocality compiler in~\cref{subsec:equiv-compiler} and our extractability compiler in~\cref{subsec:ext-compiler}. In~\cref{subsec:putting-it-together}, we briefly discuss implications for secure computation in a quantum world.
\fi

\subsection{Making Any Quantum (or Classical) Commitment Equivocal}
\label{subsec:equiv-compiler}

Recall that a quantum commitment protocol is \emph{equivocal} if an efficient quantum algorithm called the \emph{equivocator}, with access to the receiver, can generate commitments that can be opened to any value. More precisely, for any receiver (modeled as an efficient malicious quantum algorithm), there must exist an equivocator who can generate a computationally indistinguishable commitment that the equivocator can later open arbitrarily. 

In this subsection, we describe a \emph{black-box} compiler for a fairly general task (which may be of independent interest): making any \emph{classical or quantum} commitment equivocal. Recall from~\cref{subsec:overview} that we will need to invoke our equivocality compiler \emph{twice}, once on a classical bit commitment scheme, and once on an extractable quantum bit commitment scheme; in the latter case, our compiler will need to preserve the extractability of the original commitment. Since classical commitments are a subclass of quantum commitments, our exposition will focus on challenges unique to the quantum setting.  

\ifsubmission
\else

\paragraph{Why Existing Classical Solutions Are Insufficient.} Let us briefly relax our goal of using one-way functions in an exclusively black-box way. Then there is a simple equivocality compiler that applies to any (statistically-binding and computationally-hiding) classical commitment. Recall that quantum-hard one-way functions imply post-quantum zero-knowledge proofs for $\mathsf{NP}$~\cite{STOC:Watrous06}. Now, in the opening phase of the commitment, to open to a bit $b$, the committer will not send the randomness used to commit to $b$ in the clear. Instead, the committer sends a zero-knowledge proof for the \textsf{NP} statement that there exist randomness consistent with the commitment, and opening to $b$. Equivocation is achieved by simulating the zero-knowledge proof. Unfortunately, this technique fails when the commitment involves \emph{quantum} communication, since the statement to be proven is no longer an \textsf{NP} statement (nor a \textsf{QMA} statement). Therefore, we construct a compiler that only makes black-box use of the underlying commitment, and leverages Watrous's rewinding lemma in the proof of equivocality~\cite{STOC:Watrous06}. Our ideas are reminiscent of those developed in the classical setting, eg. in~\cite{TCC:PasWee09} and references therein.

\fi




\paragraph{Our Equivocality Compiler.} In our construction, to commit to a bit $b$, the committer and receiver will perform $\lambda$ sequential repetitions of the following subprotocol:
\begin{itemize}
\item The (honest) committer samples 2 uniformly random bits $u_0,u_1$, and commits to each one \emph{twice} using the base commitment scheme. Let the resulting commitments be $\mathbf{c}_{0}^{(0)},\mathbf{c}_{0}^{(1)},\mathbf{c}_{1}^{(0)},\mathbf{c}_{1}^{(1)}$, where the first two are to $u_0$ and the second two are to $u_1$. Note that since the base commitment scheme can be an arbitrary quantum interactive commitment, each commitment $\mathbf{c}_{b_1}^{(b_2)}$ corresponds to the receiver’s quantum state after the commitment phase of the base commitment.
\item The receiver sends the committer a random challenge bit $\beta$.
\item The committer opens the two base commitments $\mathbf{c}_{\beta}^{(0)},\mathbf{c}_{\beta}^{(1)}$. If the openings are invalid or the revealed messages are different, the receiver aborts the entire protocol.
\end{itemize}
If these $\lambda$ executions pass, the receiver is convinced that a majority of the committer's remaining $2\lambda$ unopened commitments are honestly generated, i.e. most pairs of commitments are to the same bit. 

Rewriting the (honest) committer’s unopened bits as $u_1,\dots,u_\lambda$, the final step of the commitment phase is for the committer to send $h_i \coloneqq u_i \oplus b$ for each $i \in [\lambda]$ (recall that $b$ is the committed bit).

To decommit, the committer reveals each $u_i$ by picking one of the two corresponding base commitments at random, and opening it. The receiver accepts if each one of the base commitment openings is valid, and the opened $u_i$ satisfies $h_i \oplus u_i = b$ for every $i$.

The (statistical) binding property of the resulting commitment can be seen to follow from the (statistical) binding of the underlying commitment. For any commitment, define the unique committed value as the majority of $(h_i \oplus u_i)$ values in the unopened commitments, setting $u_i$ to $\bot$ if both committed bits in the $i^{th}$ session differ.
Due to the randomized checks by the receiver, any committer that tries to open to a value that differs from the unique committed value will already have been caught in the commit phase, and the commitment will have been rejected with overwhelming probability. A similar argument also allows us to establish that this transformation preserves extractability of the underlying commitment. We now discuss why the resulting commitment is {\em equivocal}.

\paragraph{Quantum Equivocation.} The natural equivocation strategy should have the equivocator (somehow) end up with $\lambda$ pairs of base commitments where for each $i \in [\lambda]$, the pair of commitments is to $u_i$ and $1-u_i$ for some random bit $u_i$. This way, it can send an appropriately distributed string $h_1,\cdots,h_\lambda$, and later open to any $b$ by opening the commitment to $b \oplus h_i$ for each $i$. 

We construct our equivocator using Watrous’s quantum rewinding lemma~\cite{STOC:Watrous06} (readers familiar with Watrous’s technique may have already noticed our construction is tailored to enable its use).

We give a brief, intuition-level recap of the rewinding technique as it pertains to our equivocator. Without loss of generality, the malicious quantum receiver derives its challenge bit $\beta$ by performing some binary outcome measurement on the four quantum commitments it has just received (and on any auxiliary states). Our equivocator succeeds (in one iteration) if it can prepare four quantum commitments $\mathbf{c}_{0}^{(0)},\mathbf{c}_{0}^{(1)},\mathbf{c}_{1}^{(0)},\mathbf{c}_{1}^{(1)}$ where:
\begin{enumerate}
\item $\mathbf{c}_{\alpha}^{(0)},\mathbf{c}_{\alpha}^{(0)}$ are commitments to the same random bit,
\item $\mathbf{c}_{1-\alpha}^{(0)},\mathbf{c}_{1-\alpha}^{(0)}$ are commitments to a random bit and its complement,
\item on input $\mathbf{c}_{0}^{(0)},\mathbf{c}_{0}^{(1)},\mathbf{c}_{1}^{(0)},\mathbf{c}_{1}^{(1)}$, the receiver produces challenge bit $\beta = \alpha$.
\end{enumerate}
That is, the equivocator is successful if the receiver’s challenge bit $\beta$ corresponds to the bit $\alpha$ that it can open honestly. Watrous's~\cite{STOC:Watrous06} rewinding lemma applies if the distribution of $\beta$ is \emph{independent} of the receiver’s choice of $\alpha$, which is guaranteed here by the hiding of the base commitments. Thus, the rewinding lemma yields a procedure for obtaining an honest-looking interaction where all three properties above are met. Given the output of the rewinding process, the equivocator has successfully ``fooled’’ the committer on this interaction and proceeds to perform this for all $\lambda$ iterations. As described above, fooling the committer on all $\lambda$ iterations enables the equivocator to later open the commitment arbitrarily.


\subsection{An Extractability Compiler for Equivocal Commitments}
\label{subsec:ext-compiler}

In this subsection, we compile any classical or quantum \emph{equivocal} bit commitment into a quantum \emph{extractable} bit commitment. We stress that even though this compiler is applied to equivocal bit commitments, the resulting commitment is \emph{not} guaranteed to be simultaneously \emph{extractable and equivocal}; we refer the reader to~\cref{subsec:overview} for details on how this compiler fits into our final construction. 
%
%
Recall that a commitment scheme is \emph{extractable} if for any adversarial quantum committer that successfully completes the commitment phase, there exists an efficient quantum algorithm (called the extractor) which outputs the committed bit.

Our construction roughly follows the CK88 template for oblivious transfer, viewing the OT sender as the committer, and having the receiver use the equivocal commitment to commit to its basis choices and measurement results. We note that this is similar in spirit to a construction of statistically-hiding, computationally-binding commitment from a statistically-binding, computationally-hiding commitment (plus quantum communication) given by \cite{EC:CreLegSal01}.

\paragraph{Construction.} The committer, who intends to commit to a classical bit $b$, begins by sampling $2\lambda$-bit strings $x$ and $\theta$. It generates the corresponding $2\lambda$ BB84 states $\ket{x_i}_{\theta_i}$ and sends this to the receiver. The receiver picks $2\lambda$ random measurement bases $\hat{\theta}_i$, and measures each $\ket{x_i}_{\theta_i}$ in the corresponding basis, obtaining outcomes $\hat{x}_i$. 

Next, the receiver and committer engage in a CK88-style measurement-check subprotocol. That is, they temporarily switch roles (for the duration of the subprotocol), and perform the following steps:
\begin{enumerate}
\item The receiver (acting as a committer in the subprotocol), commits to each $\hat{\theta}_i$ and $\hat{x}_i$ (for each $i \in [2\lambda]$) with an \emph{equivocal} commitment.
\item The committer (acting as a receiver in the subprotocol), asks the receiver to open the equivocal commitments for all $i \in T$, where $T \subset [2\lambda]$ is a random set of size $\lambda$.
\item The receiver (acting as a committer in the subprotocol) opens the $\lambda$ commitments specified by $T$.
\end{enumerate}

Provided the receiver passes the measurement-check subprotocol, the committer generates the final message of the commitment phase as follows:
\begin{itemize}
\item Discard the indices in $T$ and relabel the remaining $\lambda$ indices from $1$ to $\lambda$. 
\item Partition $\{x_1,\dots,x_\lambda\}$ into $\sqrt{\lambda}$ strings $\vec{x}_1,\dots,\vec{x}_{\sqrt{\lambda}}$ each of length $\sqrt{\lambda}$. 
\item Sample $\sqrt{\lambda}$ universal hash functions $h_1,\dots,h_{\sqrt{\lambda}}$ each with 1-bit output.
\item Finally, send
$ (\theta_i)_{i \in [\lambda]}, (h_j,h_j(\vec{x}_j) \oplus b)_{j \in [\sqrt{\lambda}]}$.
\end{itemize}

This concludes the commitment phase.

To decommit, the committer reveals $b$ and $(\vec{x}_1,\dots,\vec{x}_{\sqrt{\lambda}})$. The receiver accepts if (1) for each $j$, the bit $b$ and the value $\vec{x}_j$ are consistent with the claimed value of $h_j(\vec{x}_j) \oplus b$ from the commit phase, and (2) for each index $i \in [\lambda]$ where $\theta_i = \hat{\theta}_i$, the $x_i$ from the opening is consistent with $\hat{x}_i$.

\paragraph{Extraction.} The use of equivocal commitments in the measurement-check subprotocol makes extraction simple. Given any malicious committer, we construct an extractor as follows. 

The extractor plays the role of the receiver and begins an interaction with the malicious committer. But once the committer sends its $2\lambda$ BB84 states, the extractor skips the specified measurements, instead leaving these states unmeasured. Next, instead of performing honest commitments to each $\hat{\theta}_i,\hat{x}_i$, the extractor invokes (for each commitment) the equivocator algorithm of the underlying equivocal commitment scheme. Since the equivocator is guaranteed to produce an indistinguishable commitment from the point of view of any malicious receiver for the equivocal commitment, this dishonest behavior by the extractor will go undetected. 


When the malicious committer responds with a challenge set $T \subset [2\lambda]$, the extractor samples uniformly random bases $\hat{\theta}_i$ for each $i \in T$, measures the corresponding BB84 states to obtain $\hat{x}_i$ values, and sends $(\hat{\theta}_i,\hat{x}_i)_{i \in T}$. Moreover, the equivocator (for each commitment) will enable the extractor to generate valid-looking openings for all of these claimed values. 

Thus, the malicious committer proceeds with the commitment protocol, and sends 
\[(\theta_i)_{i \in [\lambda]}, (h_j,h_j(\vec{x}_j)\oplus b)_{j \in [\sqrt{\lambda}]}\]
to the extractor. These correspond to the $\lambda$ BB84 states that the extractor has not yet measured, so it can simply read off the bases $\theta_i$, perform the specified measurements, and extract the committer's choice of $b$.


\ifsubmission
\paragraph{Statistical Hiding.}
Intuitively, statistical hiding of the above commitment protocol follows because the measurement-check subprotocol forces the receiver to measure states in arbitrary bases, which destroys information about the corresponding $x_i$ values whenever $\widehat{\theta}_i \neq \theta_i$. The formal argument is a straightforward application of a quantum sampling lemma of~\cite{C:BouFeh10}, devised in part to simplify analysis of~\cite{FOCS:CreKil88}-style protocols, and we defer further details to the full version.
\else
\paragraph{Statistical Hiding.}
Intuitively, statistical hiding of the above commitment protocol follows because the measurement-check subprotocol forces the receiver to measure states in arbitrary bases, which destroys information about the corresponding $x_i$ values whenever $\widehat{\theta}_i \neq \theta_i$. The formal argument is a straightforward application of a quantum sampling lemma of~\cite{C:BouFeh10}, devised in part to simplify analysis of~\cite{FOCS:CreKil88}-style protocols, and we defer further details to the body of the paper.
\fi

\subsection{Putting it Together: From Commitments to Secure Computation.}
\label{subsec:putting-it-together}

\ifsubmission
Plugging the compilers of~\cref{subsec:equiv-compiler,subsec:ext-compiler} into the steps described in~\cref{subsec:overview} yields a black-box construction of simultaneously extractable and equivocal quantum bit commitments from quantum-hard one-way functions. Following~\cite{C:DFLSS09}, these commitments can be plugged into CK88 to obtain maliciously simulation-secure QOT (see~\cref{sec: OT from eecom} for further details). Finally, going from QOT to arbitrary secure computation (in a black-box way) follows from prior works of~\cite{STOC:Kilian88,C:IshPraSah08,DNS12,EC:DGJMS20}; a more thorough discussion is available in the full version.
\else
Plugging the compilers of~\cref{subsec:equiv-compiler,subsec:ext-compiler} into the steps described in~\cref{subsec:overview} yields a black-box construction of simultaneously extractable and equivocal quantum bit commitments from quantum-hard one-way functions. Following~\cite{C:DFLSS09}, these commitments can be plugged into CK88 to obtain maliciously simulation-secure QOT (see~\cref{sec: OT from eecom} for further details). Finally, going from QOT to arbitrary secure computation (in a black-box way) follows from a number of prior works, which we briefly survey here for the sake of completeness (see also the discussion in~\cref{subsec: our results}).
\begin{itemize}
\item \textbf{Secure Computation of Classical Functionalities.} It is well known that maliciously simulation-secure (classical) oblivious transfer can be used in a black-box way to build two-party (classical) computation~\cite{STOC:Kilian88}. This also extends to the multi-party setting tolerating upto all-but-one corruptions~\cite{C:IshPraSah08}. Since the constructions of \cite{STOC:Kilian88} and \cite{C:IshPraSah08} only make black-box use of the underlying oblivious transfer and the simulators are straight-line, their correctness/security guarantees continue to hold if parties are quantum and the oblivious transfer uses quantum communication. Therefore, if all parties are quantum and all pairs of parties are connected via authenticated quantum channels, they can securely compute any classical functionality assuming quantum-hard one way functions (and authenticated channels in the multi-party setting).
\item \textbf{Secure Computation of Quantum Functionalities.} \cite{DNS12} constructed a secure two-party quantum computation protocol assuming black-box access to any secure two-party classical computation protocol. \cite{EC:DGJMS20} proved the analogous statement in the multi-party setting: secure multi-party classical computation implies (in a black-box way) secure multi-party quantum computation. Therefore, by invoking the results in the previous bullet, quantum-hard one-way functions also suffice for arbitrary secure quantum computation (additionally assuming authenticated channels in the multi-party setting).
\end{itemize}
While the results in the second bullet technically subsume the first, stating them separately enables direct comparison with the classical setting. In particular, one takeaway is that while oblivious transfer is necessary and sufficient for secure computation in the classical world, quantum-hard one-way functions suffice for \emph{exactly the same task} when  parties (and their communication channels) are quantum. 
\fi

\section{Preliminaries}

\paragraph{Notation.} We will write density matrices/quantum random variables (henceforth, QRVs) in lowercase bold font, e.g. $\bx$. A quantum register $\gray{X}$ will be written in uppercase (grey) serif font. A collection of (possibly entangled) QRVs will be written as $(\bx,\by,\bz)$.

Throughout this paper, $\lambda$ will denote a cryptographic security parameter. We say that a function $\mu(\lambda)$ is \emph{negligible} if $\mu(\lambda) = 1/\lambda^{\omega(1)}$.

The trace distance between two QRVs $\bx$ and $\by$ will be written as $\|\bx - \by\|_1$. Recall that the trace distance captures the maximum probability that two QRVs can be distinguished by any (potentially inefficient) procedure. We therefore say that two infinite collections of QRVs $\{\bx_\secp\}_{\secp \in \bbN}$ and $\{\by_\secp\}_{\secp \in \bbN}$ are \emph{statistically indistinguishable} if there exists a negligible function $\mu(\lambda)$ such that $||\bx_\secp - \by_\secp||_1 \leq \mu(\secp)$, and we will frequently denote this with the shorthand $\{\bx_\lambda\}_{\lambda \in \mathbb{N}} \approx_s \{\by_\lambda\}_{\lambda \in \mathbb{N}}$.



\paragraph{Non-Uniform Quantum Advice.} 
We will consider non-uniform quantum polynomial-time (QPT) algorithms \emph{with quantum advice}, denoted by $\cA = \{\cA_\secp,\brho_\secp\}_{\secp \in \bbN}$, where each $\cA_\secp$ is the classical description of a $\poly(\secp)$-size quantum circuit, and each $\brho_\secp$ is some (not necessarily efficiently computable) non-uniform $\poly(\secp)$-qubit quantum advice. We remark that ``non-uniform quantum polynomial-time algorithms'' often means non-uniform \emph{classical advice}, but the cryptographic applications in this work will require us to explicitly consider quantum advice.

\ifsubmission
\else

Therefore, \emph{computational indistinguishability} will be defined with respect to non-uniform QPT distinguishers with quantum advice. That is, two infinite collections of QRVs $\{\bx_\secp\}_{\secp \in \bbN}$ and $\{\by_\secp\}_{\secp \in \bbN}$ are computationally indistinguishable if there exists a negligible function $\mu(\cdot)$ such that for all QPT distinguishers $\cD = \{\cD_\secp,\bsigma_\secp\}_{\secp \in \bbN}$, \[\left|\Pr[\cD_\secp(\bsigma_\secp,\bx_\secp) = 1] - \Pr[\cD_\secp(\bsigma_\secp,\by_\secp) = 1] \right| \leq \mu(\secp).\]
We will frequently denote this with the shorthand $\{\bx_\secp\}_{\secp \in \bbN} \approx_c \{\by_\secp\}_{\secp \in \bbN}$.

We will also often omit the subscript ``$\secp \in \mathbb{N}$'', and simply write, for example, $\{\bx_\secp\}$.

\fi


\ifsubmission

\paragraph{Definitions for Cryptographic Commitments.} Full definitions of cryptographic commitments can be found in the full version.

\else

\subsection{Commitments}

\subsubsection{Bit Commitments}
\label{subsec: com definitions}
We define bit commitments with quantum players and quantum communication. First, we fix some notation.

A bit commitment scheme is a two-phase interactive protocol between a quantum interactive committer $\cC = (\cC_{\com},\cC_{\open})$ and a quantum interactive receiver $\cR = (\cR_{\com},\cR_{\open})$. In the commit phase, $\cC_{\com}(1^\secp,b)$ for bit $b \in \zo$ interacts with $\cR_{\com}(1^\secp)$, after which $\cC_{\com}$ outputs a state $\bx_{\com}$ and $\cR_{\com}$ outputs a state $\by_{\com}$. We denote this interaction by $(\bx_\com,\by_\com) \gets \langle \cC(1^\secp,b),\cR(1^\secp)\rangle$. 

To be precise, without loss of generality, we model each party as a sequence of unitaries and measurements acting on a local register, and a shared message register. We think of $\langle \cC(1^\secp,b),\cR(1^\secp)\rangle$ as the distribution over pairs of states $(\bx_\com,\by_\com)$ induced by the measurements of the two parties in the protocol, i.e. each sequence of measurement outcomes by the parties in the protocol defines a leftover joint pure state $(\bx_\com,\by_\com)$ (after the commit phase), with an associated probability, namely the probability of the particular sequence of measurement outcomes. 

In the open phase, $\cC_\open(\bx_{\com})$ interacts with $\cR_\open(\by_{\com})$, after which $\cR_\open$ either outputs a bit $b'$ or $\bot$. We will denote this receiver's output by $\mathsf{OUT}_\cR\langle \cC_\open(\bx_\com),\cR_\open(\by_\com)\rangle.$


First, we give the standard notions of hiding and binding. 
\begin{definition}[Hiding Commitment]
\label{def:hiding}
A bit commitment scheme is computationally (resp. statistically) \emph{hiding} if the following holds. For any polynomial-size (resp. unbounded-size) receiver $\cR_\com^* = \{\cR_{\com,\lambda}^*,\brho_\lambda\}$ interacting in the commit phase of the protocol, let $\mathsf{OUT}_\cR\langle \cC_\com(1^\lambda,b),\cR_{\com,\lambda}^*(\brho_\lambda)\rangle$ denote a bit output by $\cR_{\com,\lambda}^*$ after interaction with an honest $\cC_\com$ committing to message $b$. Then for every polynomial-size (resp. unbounded-size) receiver $\cR_\com^* = \{\cR_{\com,\lambda}^*,\brho_\lambda\}$, there exists a negligible function $\nu(\cdot)$ such that $$\left|\Pr[\mathsf{OUT}_\cR\langle\cC_\com(1^\lambda,0),\cR_{\com,\lambda}^*(\brho_\lambda)\rangle = 1] - \Pr[\mathsf{OUT}_\cR\langle\cC_\com(1^\lambda,1),\cR_{\com,\lambda}^*(\brho_\lambda)\rangle = 1]\right| = \nu(\lambda).$$
\end{definition}


\begin{definition}[Statistically Binding Commitment]
\label{def:binding}
A bit commitment scheme is statistically binding if for every unbounded-size committer $\cC^* = (\cC^*_{\com}, \cC^*_{\open})$, there exists a negligible function $\nu(\cdot)$ such that with probability at least $1-\nu(\secp)$ over $(\bx_\com,\by_\com) \gets \langle \cC^*_{\com},\cR(1^\lambda) \rangle$, there exists a bit $b \in \{0,1\}$ such that $$\Pr[\mathsf{OUT}_\cR\langle \cC^*_{\open}(1^\lambda, \bx_{\com}),\cR_\open(1^\lambda,\by_\com)\rangle = b] \leq \nu(\lambda).$$
\end{definition}

As mentioned above, $\langle \cC(1^\secp,b),\cR(1^\secp)\rangle$ is the distribution over pairs of states $(\bx_\com,\by_\com)$ induced by the measurements performed by the parties in the protocol. Thus, when we write ``with probability $1-\nu(\lambda)$ over $(\bx_\com,\by_\com) \gets \langle \cC^*_{\com},\cR(1^\lambda) \rangle$”, this probability is taken over the measurement outcomes of the two parties in the protocol. Our definition then says that with overwhelming probability over such sampling, there exists a bit $b$ such that the malicious committer, given $\bx_\com$ in the open phase, can only open to $b$ with negligible probability.

We will not consider (plain) computationally binding commitments in this work. We remark that there are subtleties to consider when defining computationally binding commitment secure against quantum attackers, and that the ``right'' definition appears to be Unruh's notion of collapse-binding~\cite{EC:Unruh16}. However, for this work, computational notions of binding will only arise in the context of \emph{extractable commitments}, which can be seen as a strong form of computationally binding commitments.

Post-quantum bit commitments satisfying computational hiding (as in Definition \ref{def:hiding}) and statistical binding (as in Definition \ref{def:binding}) can be constructed from any post-quantum pseudorandom generator (PRG)~\cite{JC:Naor91}.\footnote{In Naor's protocol, the receiver first sends a uniformly random $u \gets \{0,1\}^{\lambda^3}$, and the committer commits to bit $b$ by sending $(b \cdot u) \oplus G(s)$, where $s \gets \{0,1\}^\secp$ and $G$ is a length-tripling PRG.} Watrous~\cite{STOC:Watrous06} considers PRGs built from (post-quantum) one-way \emph{permutations} via a ``quantum Goldreich-Levin Theorem'' of~\cite{AdcockCleve02}. However, Aaronson~\cite{Aaronson09} and Zhandry~\cite{FOCS:Zhandry12} later pointed out that the original~\cite{HILL99} construction of PRGs from one-way \emph{functions}  extends to non-uniform quantum adversaries with \emph{classical} advice. This can be extended to handle non-uniform {quantum} advice by giving the one-way function attacker constructed in the~\cite{HILL99} reduction many copies of the PRG attacker's non-uniform quantum advice (which only requires some polynomial upper bound on the number of times the reduction invokes the PRG attacker). Thus, post-quantum computationally hiding and statistically binding  commitments are known from post-quantum one-way functions.
(We note that statistically binding commitments can alternatively be implemented non-interactively with quantum communication from one-way functions~\cite{ywlq}.)

We remark that achieving general definitions for quantum binding is a notoriously subtle task that has been studied in many prior works, and that our definition is more restrictive than other definitions in the literature (e.g. \cite{FUYZ20}). However, our goal is to give a definition that (i) we can satisfy and (ii) suffices for our eventual construction of Quantum Oblivious Transfer. Definition \ref{def:binding} directly captures what it means for a quantum bit commitment to be statistically binding to a \emph{classical} bit.

Next, we define notions of extractability and equivocation for commitments.

\subsubsection{Extractable Commitments}
Informally, a commitment is extractable if for every QPT adversarial committer $\cC^*$, there exists a QPT {\em extractor} $\cE_{\cC^*}$ that in an ideal world extracts the message committed by $\cC^*$ while also simulating $\cC^*$'s view. 
Looking ahead, extractable commitments will enable a challenger that runs an experiment against an adversary to extract messages committed by the adversary without having to execute the opening phase.
We formally define this notion below.

\begin{definition}[Extractable Commitment]\label{def:extcom}
Let $\{\brho_{\lambda},\bsigma_{\lambda}\}_{\lambda \in \mathbb{N}}$ denote auxiliary quantum states that may be entangled.
A bit commitment scheme is \emph{extractable} if for every polynomial-size quantum adversarial committer $\cC^* = \{\cC^*_{\com,\lambda}, \brho_{\lambda}\}_{\lambda \in \mathbb{N}}$, there exists a QPT extractor $\cE_{\cC^*_{\com}}$  such that for every polynomial-size adversarial opening strategy $\{\cC^*_{\open,\lambda}\}_{\lambda \in \mathbb{N}}$ and every polynomial-size quantum distinguisher $\cD^* = \{\cD^*_{\lambda}, \bsigma_{\lambda}\}_{\lambda \in \mathbb{N}}$, there exists a negligible function $\nu(\cdot)$ such that:
$$\big| \Pr[\cD^*_\lambda(\bsigma_\lambda,\mathsf{Real})=1] - \Pr[\cD^*_\lambda(\bsigma_\lambda,\mathsf{Ideal}) = 1] \big| = \nu(\lambda)$$ 
for $\mathsf{Real}$ and $\mathsf{Ideal}$ distributions defined below (where we omit indexing by $\lambda$ for the sake of presentation).
\begin{itemize}
    \item Denote by $\mathsf{Real}$ the distribution consisting of $(\bx_{\final},b)$ where $\bx_{\final}$ denotes the final state of $\cC^*_\open$ and $b \in \{0,1,\bot\}$ denotes the output by $\cR_\open$ after the open phase.
    \item To define the distribution $\mathsf{Ideal}$, first run the extractor $\cE_{\cC^*_{\com}}$ on input $\brho$ to obtain outputs $(\bx_{\com},\by_{\com},b^*)$, where $\bx_{\com}$ and $\by_{\com}$ are the final states of committer and receiver respectively after the commit phase. Then run $\cC^*_{\open}(\bx_{\com})$ and $\cR(\by_{\com})$ to produce $(\bx_{\final},b)$, where $\bx_{\final}$ is the final state of the committer. If $b \notin \{\bot,b^*\}$, output $\mathsf{FAIL}$ and otherwise output $(\bx_{\final},b)$.
\end{itemize}
\end{definition}

\subsubsection{Equivocal Commitments}
Informally, a commitment is equivocal if for every QPT adversarial receiver $\cR^*$, there exists a QPT {\em equivocator} $\cQ_{\cR^*}$ that generates a view for $\cR^*$ in the commit phase {\em without any input}, and then receives an input bit $b$ once the commit phase ends. The equivocator is then able to open the commitment generated in the commit phase to bit $b$ such that: $\cR^*$ cannot distinguish between its interaction with $\cQ_{\cR^*}$ and its interaction with an honest committer that committed to $b$ all along, and opened its commitment honestly in the opening phase.
Looking ahead, equivocal commitments will enable a challenger that runs an experiment against an adversary to generate commitments that can later be opened arbitrarily. We will now formally describe this primitive.

\begin{definition}[Equivocal Commitments]
\label{def:eqcom}
Let $\{\brho_{\lambda},\bsigma_{\lambda}\}_{\lambda \in \mathbb{N}}$ denote auxiliary quantum states that may be entangled.
A bit commitment scheme is {\em equivocal} if for every polynomial-size quantum receiver 
$\cR^* = \{\cR^*_{\com,\lambda},\cR^*_{\open,\lambda},\brho_\lambda\}$
there exists a QPT equivocator $\equivsim_{\cR^*} = (\equivsim_{\cR^*,\com}, \equivsim_{\cR^*,\open})$ such that 
for every polynomial-size quantum distinguisher $\cD^* = \{\cD^*_{\lambda}, \bsigma_{\lambda}\}$ there exists a negligible function $\nu(\cdot)$ such that for every bit $b \in \zo$,
$$\big| \Pr[\cD^*(\bsigma_\lambda,\mathsf{Real}_b)=1] - \Pr[\cD^*(\bsigma_\lambda,\mathsf{Ideal}_b) = 1] \big| = \nu(\lambda)$$ 
for $\mathsf{Real}_b$ and $\mathsf{Ideal}_b$ distributions as defined below (where we omit indexing by $\lambda$ for the sake of presentation).
\begin{itemize}
    \item $\mathsf{Real}_b$: 
    Let $\cR^*_\com(\brho)$ interact with $\cC_{\com}(1^\lambda,b)$ in the commit phase, after which $\cC_{\com}$ obtains state $\bx_\com$ and $\cR^*$ obtains state $\by_\com$. Then, let $\cR^*_\open(\by_\com)$ interact with $\cC_\open(\bx_\com)$ and output the final state $\by_\final$ of $\cR^*_\open$.

    \item $\mathsf{Ideal}_b$: First run the equivocator
    $\equivsim_{\cR^*,\com}(\brho)$ to obtain state $(\bx_\com$, $\by_\com)$. Then, let $Q_{\cR^*,\open}(b,\bx_\com)$ interact with $\cR^*_\open(\by_{\com})$ and output the final state $\by_\final$ of $\cR^*_\open$.
\end{itemize}

\end{definition}

\fi

\ifsubmission
\else
\subsection{Quantum Rewinding Lemma}

We will make use of the following lemma from~\cite{STOC:Watrous06}.

\begin{lemma}\label{lem:qrl} Let $\cQ$ be a general quantum circuit with $n$ input qubits that outputs a classical bit $b$ and $m$ qubits. For an $n$-qubit state $\ket{\psi}$, let $p(\ket{\psi})$ denote the probability that $b =0$ when executing $\cQ$ on input $\ket{\psi}$. Let $p_0, q \in (0,1)$ and $\epsilon \in (0,1/2)$ be such that:

\begin{itemize}
    \item For every $n$-qubit state $\ket{\psi},p_0 \leq p(\psi)$,
    \item For every $n$-qubit state $\ket{\psi}$, $|p(\psi)-q| < \epsilon$,
    \item $p_0(1-p_0) \leq q(1-q)$,
\end{itemize}

Then, there is a general quantum circuit $\widehat{\cQ}$ of size $O\left(\frac{\log(1/\epsilon)}{4 \cdot p_0(1-p_0)} |\cQ|\right)$, taking as input $n$ qubits, and returning as output $m$ qubits, with the following guarantee. For an $n$ qubit state $\ket{\psi}$, let $\cQ_0(\ket{\psi})$ denote the output of $\cQ$ on input $\ket{\psi}$ conditioned on $b=0$, and let $\widehat{\cQ}(\ket{\psi})$ denote the output of $\widehat{\cQ}$ on input $\ket{\psi}$. Then, for any $n$-qubit state $\ket{\psi}$,

$$\mathsf{TD}\left(\cQ_0(\ket{\psi}),\widehat{\cQ}(\ket{\psi})\right) \leq 4\sqrt{\epsilon}\frac{\log(1/\epsilon)}{p_0(1-p_0)}.$$

\end{lemma}
\fi

\ifsubmission
\else
\subsection{Quantum Entropy and Leftover Hashing}
\label{subsec: quantum min entropy}

\paragraph{Classical Min-Entropy.}For a classical random variable $X$, its min-entropy $\mathbf{H}_\infty(X)$ is defined as
\[ \mathbf{H}_\infty(X) \coloneqq -\log(\max_{x} \Pr[X = x]).\]

In cryptographic settings, we are often interested in the min-entropy of a random variable $X$ sampled from a joint distribution $(X,Y)$, where $Y$ is side information available to an adversary/distinguisher. Following~\cite{AC:RenWol05}, we define the conditional min-entropy of $X$ given $Y$ as
\[ \mathbf{H}_\infty(X \mid Y) \coloneqq -\log(\max_{x,y}\Pr[X = x \mid Y = y]).\]

That is, $\mathbf{H}_\infty(X \mid Y)$ is (the negative log of) the maximum probability of guessing the outcome of $X$, maximized over the possible outcomes of $Y$.

\fi





\ifsubmission
\else

\paragraph{Quantum Conditional Min-Entropy.} 

Let $\boldsymbol{\rho}_{\gray{X}\gray{Y}}$ denote a bipartite quantum state over registers $\gray{X}\gray{Y}$. Following~\cite{Renner08,KonRenSch09}, the conditional min-entropy of $\boldsymbol{\rho}_{\gray{X}\gray{Y}}$ given $\gray{Y}$ is then defined to be 
\[\mathbf{H}_\infty(\boldsymbol{\rho}_{\gray{X}\gray{Y}} \mid \gray{Y}) \coloneqq \sup_\mathbf{y} \max \{h \in \mathbb{R} : 2^{-h} \cdot I_{\gray{X}} \otimes \mathbf{y}_{\gray{Y}} - \boldsymbol{\rho}_{\gray{X}\gray{Y}} \geq 0\}.\]

In this work, we will exclusively consider the case where the $\boldsymbol{\rho}_{\gray{X}\gray{Y}}$ is a joint distribution of the form $(X,\mathbf{y})$ where $X$ is a classical random variable. In other words, $\boldsymbol{\rho}_{\gray{X}\gray{Y}}$  can be written as 
\[ \sum_x \Pr[X = x] \ket{x}\bra{x} \otimes \mathbf{y}_x. \]

In this case, we will write $\mathbf{H}_\infty(\boldsymbol{\rho}_{\gray{X}\gray{Y}} \mid \gray{Y})$ as $\mathbf{H}_\infty(X \mid \mathbf{y})$. We remark that in this particular setting, $\mathbf{H}_\infty(X \mid \mathbf{y})$ can be interpreted as the (negative log of) the maximum probability of guessing $X$ given quantum state $\mathbf{y}$~\cite{KonRenSch09}.

\paragraph{Leftover Hash Lemma with Quantum Side Information.} We now state a generalization of the leftover hash lemma to the setting of quantum side information. 
\begin{lemma}[\cite{TCC:RenKon05}]\label{thm:privacy-amplification}
Let $\mathcal{H}$ be a family of universal hash functions from $\mathcal{X}$ to $\{0,1\}^\ell$, i.e. for any $x \neq x'$, $\Pr_{h \leftarrow \mathcal{H}}[h(x) = h(x')] = 2^{-\ell}$. Then for joint random variables $(X,\mathbf{y})$ where $X$ is a classical random variable over $\mathcal{X}$ and $\mathbf{y}$ is a quantum random variable,
\[\left\|(h,h(X),\mathbf{y}) - (h,u,\mathbf{y})\right\|_1 \leq \frac{1}{2^{1+\frac{1}{2}(\mathbf{H}_\infty(X|\mathbf{y})-\ell)}},\]
where $h \leftarrow \mathcal{H}$ and $u \leftarrow \{0,1\}^\ell$.

\end{lemma}

\fi

\renewcommand{\extcommit}{\mathsf{Commit}}
\renewcommand{\extdecommit}{\mathsf{Decommit}}

\section{A Quantum Equivocality Compiler}
\label{sec: equivocal commitment}

In this section, we show a generic black-box compiler that takes any quantum-secure bit commitment scheme and produces a quantum-secure \emph{equivocal} bit commitment scheme. 

\ifsubmission
The compiler is described in \proref{fig:eqextcom}, where $(\extcommit,\extdecommit)$ denotes some statistically binding and computationally hiding bit commitment scheme. We describe how to equivocally commit to a single bit, and note that commitment to an arbitrary length string follows by sequential repetition. 
\else
The compiler is described in \proref{fig:eqextcom}, where $(\extcommit,\extdecommit)$ denotes some statistically binding and computationally hiding bit commitment scheme satisfying Definitions \ref{def:hiding} and \ref{def:binding}. We describe how to equivocally commit to a single bit, and note that commitment to an arbitrary length string follows by sequential repetition. 
\fi


\ifsubmission
Furthermore, we show that if the underlying commitment $(\extcommit,\extdecommit)$ is \emph{extractable}, then the resulting commitment is both extractable and equivocal.
\else
Furthermore, we show that if the underlying commitment $(\extcommit,\extdecommit)$ is \emph{extractable} according to~\cref{def:extcom}, then the resulting commitment is both extractable and equivocal (\cref{def:eqcom}).
\fi

These results are captured in the following theorems.


\renewcommand{\view}{\rho}

\protocol
{\proref{fig:eqextcom}}
{Equivocal Bit Commitment.}
{fig:eqextcom}
{
\textbf{Committer $\cC$ Input:} Bit $b \in \{0,1\}$.

\underline{\textbf{The Protocol: Commit Phase}}
\begin{enumerate}
\item $\cC$ samples uniformly random bits $d_{i,j}$ for $i \in [\lambda]$ and $j \in \zo$.
\item For every $i \in [\lambda]$, $\cC$ and $\cR$ sequentially perform the following steps.
\begin{enumerate}
\item $\cC$ and $\cR$ execute four sessions sequentially, namely:
\begin{itemize}
    \item $\bx_{0,0}, \by_{0,0} \leftarrow \extcommit \langle \cC(d_{i,0}), \cR \rangle$,
    \item $\bx_{0,1}, \by_{0,1} \leftarrow \extcommit \langle \cC(d_{i,0}), \cR \rangle$, 
    \item $\bx_{1,0}, \by_{1,0} \leftarrow \extcommit \langle \cC(d_{i,1}), \cR \rangle$ and 
    \item $\bx_{1,1}, \by_{1,1} \leftarrow \extcommit \langle \cC(d_{i,1}), \cR \rangle$.
\end{itemize}
\item $\cR$ sends a choice bit $c_i \leftarrow \zo$.
\item $\cC$ and $\cR$ execute two decommitments, obtaining the opened bits:
\begin{itemize}
\item $u \leftarrow \extdecommit \langle \cC(\bx_{c_i,0}), \cR(\by_{c_i,0}) \rangle$ and
\item $v \leftarrow \extdecommit \langle \cC(\bx_{c_i,1}), \cR(\by_{c_i,1}) \rangle$. 
\end{itemize}
 If $u \neq v$, $\cR$ aborts.
 Otherwise, $\cC$ and $\cR$ continue.
\end{enumerate}
\item For $i \in [\lambda]$, 
$\cC$ sets $e_i = b \oplus d_{i, 1-c_i}$ and
sends $\{e_i\}_{i \in [\lambda]}$ to $R$.
\end{enumerate}

\underline{\textbf{The Protocol: Decommit Phase}}
\begin{enumerate}
    \item $\cC$ sends $b$ to $\cR$. In addition,
    \begin{itemize}
    \item For $i \in [\lambda]$, 
    $\cC$ picks $\alpha_i \gets \zo$ and sends it to $\cR$. 
    \item $\cC$ and $\cR$ execute  $\widehat{d}_i \leftarrow \extdecommit \langle \cC(\bx_{1-c_i,\alpha_i}), \cR(\by_{1-c_i,\alpha_i}) \rangle$.
    \end{itemize}
    \item $\cR$ accepts the decommitment and outputs $b$ if for every $i \in [\lambda]$, $\widehat{d}_i = b \oplus e_i$.
\end{enumerate}
}



\ifsubmission
\begin{theorem}
\label{thm:equiv-compiler}
For $\mathcal{X} \in \{\text{quantum extractability},\text{statistical binding}\}$ and $\mathcal{Y} \in$ \\ $\{\text{computationally},\text{statistically}\}$, if $\extcommit$ is a $\mathcal{Y}$-hiding quantum bit commitment satisfying $\mathcal{X}$, then \proref{fig:eqextcom} is a $\mathcal{Y}$-equivocal bit commitment satisfying $\mathcal{X}$. 
\end{theorem}
\else
\begin{theorem}
\label{thm:equiv-compiler}
For $\mathcal{X} \in \{\text{quantum extractability},\text{statistical binding}\}$ and $\mathcal{Y} \in \{\text{computationally},\text{statistically}\}$, if $\extcommit$ is a $\mathcal{Y}$-hiding quantum bit commitment satisfying $\mathcal{X}$, then \proref{fig:eqextcom} is a $\mathcal{Y}$-equivocal bit commitment satisfying $\mathcal{X}$. 
\end{theorem}
\fi

\ifsubmission
\else
Definitions of statistical binding and computational hiding for quantum commitments can be found in~\cref{def:binding} and~\cref{def:hiding}. Definitions of quantum extractability and quantum equivocality are given in~\cref{def:extcom} and~\cref{def:eqcom}.
\fi


These theorems follow from establishing statistical binding, equivocality, and extractability of the commitment in \proref{fig:eqextcom}, as we do next.
First, we note that if $\mathsf{Commit}$ is statistically binding, then \proref{fig:eqextcom} is statistically binding. For any adversarial committer strategy, consider the $\secp$ unopened pairs of commitments after the commit phase. Since $\mathsf{Commit}$ is statistically binding, we can assume that each of the $2\secp$ commitments is binding to a particular bit, except with negligible probability. Now, if any single pair contains binding commitments to the same bit $d_{i}$, then the committer will only be able to open its \proref{fig:eqextcom} commitment to the bit $d_i \oplus e_i$. Thus, to violate binding, the adversarial committer will have to have committed to different bits in each of the $\secp$ unopened pairs. However, in this case, the committer will be caught and the receiver will abort except with probability $1/2^\secp$.

\subsection{Equivocality}
\ifsubmission
The equivocal simulator $(\equivsim_{\cR^*,\com}, \equivsim_{\cR^*,\open})$ is obtained via the use of Watrous's quantum rewinding lemma~\cite{STOC:Watrous06}; a full statement of the lemma is available in the full version. For the purposes of defining the simulation strategy, it will be sufficient (w.l.o.g.) to consider a restricted receiver $\cR^*$ as follows, for the $i^{th}$ sequential step of the protocol.
In our simulation, the state of $\cR^*$ will be initialized to the final state at the end of simulating the $(i-1)^{th}$ step.
\else
The equivocal simulator $(\equivsim_{\cR^*,\com}, \equivsim_{\cR^*,\open})$ is obtained via the use of the quantum rewinding lemma (Lemma \ref{lem:qrl})~\cite{STOC:Watrous06}.
For the purposes of defining the simulation strategy, it will be sufficient (w.l.o.g.) to consider a restricted receiver $\cR^*$ as follows, for the $i^{th}$ sequential step of the protocol.
In our simulation, the state of $\cR^*$ will be initialized to the final state at the end of simulating the $(i-1)^{th}$ step.
\fi
\begin{enumerate}
\item $\cR^*$ takes a quantum register $\mathsf{W}$, representing its auxiliary quantum input. $\cR^*$ will use two additional quantum registers that function as work space: $\mathsf{V}$, which is an arbitrary (polynomial-size) register, and $\mathsf{A}$, which is a single qubit register. The registers $\mathsf{V}$ and $\mathsf{A}$ are initialized to their all-zero states before the protocol begins.

\item Let
$\textsf{M}$ denote the polynomial-size register used by $\cC$ to send messages to $\cR^*$. After carrying out step 2(a) by running on registers $(\mathsf{W},\mathsf{V},\mathsf{A},\mathsf{M})$, $\cR^*$ measures the register $\mathsf{A}$ to obtain a bit $c_i$, for Step 2(b), which it sends back to $\cC$.

\item Next, $\cR^*$ computes the decommitment phases (with messages from $\cC$ placed in register $\mathsf{M}$) according to Step 2(c). $\cR^*$ outputs registers $(\mathsf{W},\mathsf{V},\mathsf{A}, \mathsf{M})$.
\end{enumerate}

Any polynomial-time quantum receiver can be modeled as a receiver of this restricted form followed by some polynomial-time post-processing of the restricted receiver’s output. The same post-processing can be applied to the output of the simulator that will be constructed for the given restricted receiver.


Following~\cite{STOC:Watrous06}, we define a simulator that uses two additional registers, $\mathsf{C}$ and $\mathsf{Z}$.
$\mathsf{C}$ is a one qubit register, while $\mathsf{Z}$ is an auxiliary register used to implement the computation that will be described next.
Consider a quantum procedure $\equivsim_{\partial}$ that implements the strategy described in  \proref{fig:eqsimulator} using these registers.

\protocol
{\proref{fig:eqsimulator}}
{Equivocal Simulator.}
{fig:eqsimulator}
{

\underline{\textbf{Circuit $\equivsim_{\partial}$.}}
\begin{enumerate}
\item Sample a uniformly random classical bit $\widehat{c}$, and store it in register $\mathsf{C}$.
\item Sample uniformly random bits $(z,d)$.
\item 
If $\widehat{c} = 0$, initialize committer input as follows, corresponding to four sequential sessions:
\begin{itemize}
    \item For the first two sessions, set committer input to $z$.
    \item For the third and fourth sessions, set committer input to $d$ and $1-d$ respectively.
\end{itemize}
\item 
If $\widehat{c} = 1$, initialize committer input as follows, corresponding to four sequential sessions:
\begin{itemize}
    \item For the first and second sessions, set committer input to $d$ and $1-d$ respectively.
    \item For the last two sessions, set committer input to $z$.
\end{itemize}
\item Run the commitment phase interaction between the honest committer and $\cR^*$'s sequence of unitaries on registers initialized as above.
\item Measure the qubit register $A$ to obtain a bit $c$. If $c=\widehat{c}$, output 0, otherwise output $1$.
\end{enumerate}
}

Next, we would like to apply Watrous's quantum rewinding lemma to the $\equivsim_{\partial}$ circuit.
In order to do this, we will argue that the probability $p(\psi)$ that this circuit outputs $0$ is such that $|p(\psi)-\frac{1}{2}| = \mathsf{negl}(\lambda)$, regardless of the auxiliary input $\ket{\psi}$ to $\cR^*$. 
This follows from the fact that the commitments are (statistically/computationally) hiding.
In more detail, by definition, Step 5 produces a distribution on the $\cR^*$'s side that is identical to the distribution generated by $\cR^*$ in its interaction with the committer. 
If $|p(\psi)-\frac{1}{2}|$ were non-negligible, then the sequence of unitaries applied by $\cR^*$ could be used to distinguish commitments generated according to the case $\widehat{c} = 0$ from commitments generated according to the case $\widehat{c} = 1$, leading to a contradiction.

Now consider the state of the residual qubits of $\equivsim_{\partial}$ conditioned on a measurement of its output qubit being $0$. The output state of the general quantum circuit $\widehat{\equivsim}$ resulting from applying Watrous's quantum rewinding lemma will have negligible trace distance from this state. This state is over all of the registers discussed above, so the simulator $\equivsim_{\com,\cR^*}$ must further process this state as:
\begin{itemize}
\item Measure the register $\mathsf{C}$, obtaining challenge $c$.
\item Compute decommitment information corresponding to challenge $c$, as in Step 2(c) of the protocol (recall that this information is stored in the message register $\mathsf{M}$). 
\item Output registers $(\mathsf{W}, \mathsf{V}, \mathsf{A}, \mathsf{M})$.  All remaining registers are traced out.
\end{itemize}


The simulator $\equivsim_{\cR^*,\com}$ executes all $\lambda$ sequential interactions in this manner, and then
 samples $e_1, \ldots, e_{\lambda} \leftarrow \{0,1\}^{\lambda}$, as the committer messages for Step 3 of \proref{fig:eqextcom}. It runs the receiver's unitary on the resulting protocol, and outputs the resulting registers $(\mathsf{W}, \mathsf{V}, \mathsf{A}, \mathsf{M})$.
It additionally outputs private state $\mathsf{st} = (c_1,d_1,\ldots, c_\lambda,d_\lambda)$ where $c_i,d_i$ were sampled during the $i$th execution of \proref{fig:eqsimulator}. 

The simulator $\equivsim_{\cR^*,\open}(b,\mathsf{st},\bw, \bv, \ba, \bm)$ 
parses $\mathsf{st}$ as $(c_1,d_1,\ldots, c_\lambda,d_\lambda)$.
For every $i \in [\lambda]$ it does the following:\begin{itemize}
\item Let $\widehat{d}_i = b \oplus e_i$.
\item If $c_i = 0$, it executes the decommitment phase for the $((\widehat{d}_i \oplus d_i)+2)^{th}$ session.
\item If $c_i = 1$, it executes the decommitment phase for the $(\widehat{d}_i \oplus d_i)^{th}$ session.
\end{itemize}
$\equivsim_{\cR^*,\open}$ then executes the receiver's algorithm on these decommitments and outputs the resulting state. Note that each decommitment will be to the bit $\widehat{d}_i = b \oplus e_i$.

\ifsubmission
To complete the proof of equivocality, we must establish that the view of the receiver interacting with an honest committer the view of the receiver interacting with the equivocator are indistinguishable. This follows from the (statistical/computational) hiding of the commitment scheme, via an identical argument to the one used above. In particular, if the equivocal simulator produces a distribution that is distinguishable from the real distribution, then there exists a session $i \in [\lambda]$ such that the distribution in the real and ideal experiments upto the ${i-1}^{th}$ session are indistinguishable, but upto the $i^{th}$ session are distinguishable. This contradicts the above guarantee given by the quantum rewinding lemma, since for any $i$, the post-processed residual qubits of $\equivsim_{\partial}$ are indistinguishable from the state of $\cR^*$ after the $i^{th}$ sequential session in the real protocol (due to the hiding of the commitment scheme).
\else
To complete the proof of equivocality, we must establish that the distributions $\mathsf{Real}_b$ (the view of the receiver interacting with an honest committer) and $\mathsf{Ideal}_b$ (the view of the receiver interacting with the equivocator) of Definition \ref{def:eqcom}, are indistinguishable. This follows from the (statistical/computational) hiding of the commitment scheme, via an identical argument to the one used above. In particular, if the equivocal simulator produces a distribution that is distinguishable from the real distribution, then there exists a session $i \in [\lambda]$ such that the distribution in the real and ideal experiments upto the ${i-1}^{th}$ session are indistinguishable, but upto the $i^{th}$ session are distinguishable. This contradicts the above guarantee given by the quantum rewinding lemma, since for any $i$, the post-processed residual qubits of $\equivsim_{\partial}$ are indistinguishable from the state of $\cR^*$ after the $i^{th}$ sequential session in the real protocol (due to the hiding of the commitment scheme).
\fi

\subsection{Extractability}

Next, we prove that \proref{fig:eqextcom} satisfies extractability as long as the underlying commitment $(\extcommit,\extdecommit)$ is extractable; in other words, this compiler preserves extractability. Consider the following extractor $\cE_{\cC^*}$.

\begin{itemize}
    \item For $i \in [\secp]$:
    \begin{itemize}
        \item Execute four sequential commitment sessions with $\cC^*$, where the extractor of $\extcommit$ is run on all sessions. Obtain  outputs $(\brho_{\cC^*}, \st_{\cR,i,0}, d'_{i,0}, \st_{\cR,i,1}, d'_{i,1})$, where $\brho_{\cC^*}$ is the final state of the committer after engaging in all four sequential sessions, and $\st_{\cR,i,0}$, $\st_{\cR,i,1}$ are receiver states output by the extractor corresponding to the first and third sessions. 
        \item Corresponding to Step 2(b), compute and send $c_i \leftarrow \zo$.
        \item Execute Step 2(c) identically to \proref{fig:eqextcom}.
    \end{itemize}
    \item Executes Step 3 of \proref{fig:eqextcom}, receiving bits $\{e_i\}_{i \in [\secp]}$. Fix $b^*$ to be the most frequently ocurring bit in $\{e_i \oplus d'_{i,1-c_i}\}_{i \in [\secp]}$, and output the final state of $\cC^*$, the receiver states $\{\st_{\cR,i,0},\st_{\cR,i,1}\}_{i \in [\secp]}$, and the extracted bit $b^*$.
\end{itemize}

Since the extractor is applied sequentially to parts of the committer, we have that the extractor runs in time polynomial in $\lambda$.
Next, we will prove indistinguishability between the distributions $\mathsf{Real}$ and $\mathsf{Ideal}$ defined by the above extractor via a hybrid argument, and based on the definition of extractability of the underlying commitment $(\extcommit, \extdecommit)$. 
In more detail, recall that $\mathsf{Real}$ denotes the distribution $(\rho_{\cC^*,\final}, b)$ where $\rho_{\cC^*,\final}$ denotes the final state of $\cC^*$ and $b$ the output of the receiver, and $\mathsf{Ideal}$ denotes the final committer state and opened bit after the opening phase of the scheme is run on the output of the extractor.

Note that there are a total of $4\lambda$ commitment sessions. Denote the real experiment by $\mathsf{Hybrid}_{0,1,1}$.
For each $i \in [\lambda], j \in [0,1], k \in [0,1]$, define $\mathsf{Hybrid}_{i,j,k}$ to be the distribution obtained as follows.

\paragraph{Commit Phase.}
Set $a = 1$ and do the following:
\begin{enumerate}
    \item If $a = \lambda + 1$, obtain $\{e_i\}_{i \in [\lambda]}$ from $\cC^*$ and end.
    \item If $a < i$,
    \begin{enumerate}
    \item Execute four sessions sequentially, namely:
    \begin{itemize}
        \item $(\bx_{0,0}, \by_{0,0}, d_{0,0}^{(a)}) \leftarrow \cE_{\cC^*_{i,0,0}}(\brho)$
        \item $(\bx_{0,1}, \by_{0,1}, d_{0,1}^{(a)}) \leftarrow \cE_{\cC^*_{i,0,1}}(\bx_{0,0})$
        \item $(\bx_{1,0}, \by_{1,0}, d_{1,0}^{(a)}) \leftarrow \cE_{\cC^*_{i,1,0}}(\bx_{0,1})$
        \item  
        $(\bx_{1,1}, \by_{1,1}, d_{1,1}^{(a)}) \leftarrow \cE_{\cC^*_{i,1,1}}(\bx_{1,0})$
    \end{itemize}
    \item Send choice bit $c_a \leftarrow \{0,1\}$ to $\cC^*$.
    \item Execute two decommitments, obtaining the opened bits:
    \begin{itemize}
        \item $(\bx', u_a) \leftarrow \extdecommit \langle \cC^*(\bx_{1,1}), \cR(\by_{c_a,0}) \rangle$ and
        \item $(\bx'',v_a) \leftarrow \extdecommit \langle \cC^*(\bx'), \cR(\by_{c_a,1}) \rangle$. 
        \item Abort if $u_a \neq v_a$.
    \end{itemize}
    \item Update state $\brho = \bx''$, and set $a = a+1$.
    \end{enumerate}
    \item If $a = i$, 
    \begin{enumerate}
    \item Execute four sessions sequentially, namely:
    \begin{itemize}
        \item 
        Set $(\bx_{0,0}, \by_{0,0}, d_{0,0}^{(a)}) \leftarrow \cE_{\cC^*_{i,0,0}}(\brho)$.
        \item 
        If $(j,k) = (0,0)$, set $(\bx_{0,1}, \by_{0,1}) \leftarrow \extcommit \langle \cC^*(\bx_{0,0}), \cR \rangle$ and $d_{0,1}^{(a)} = \tilde{\bot}$.
        
        Else set $(\bx_{0,1}, \by_{0,1}, d_{0,1}^{(a)}) \leftarrow \cE_{\cC^*_{i,0,1}}(\bx_{0,0})$.
        \item If $j = 0$, set $\bx_{1,0}, \by_{1,0} \leftarrow \extcommit \langle \cC^*(\bx_{0,1}), \cR \rangle$ and $d_{1,0}^{(a)} = \tilde{\bot}$.
        
        Else set $(\bx_{1,0}, \by_{1,0}, d_{1,0}^{(a)}) \leftarrow \cE_{\cC^*_{i,1,0}}(\bx_{0,1})$.
        \item  If $j = 0$ or $k = 0$, set $\bx_{1,1}, \by_{1,1} \leftarrow \extcommit \langle \cC^*(\bx_{1,0}), \cR \rangle$ and $d_{1,1}^{(a)} = \tilde{\bot}$.
        
        Else set $(\bx_{1,1}, \by_{1,1}, d_{1,1}^{(a)}) \leftarrow \cE_{\cC^*_{i,1,1}}(\bx_{1,0})$.
    \end{itemize}
    \item Send choice bit $c_a \leftarrow \{0,1\}$ to $\cC^*$.
    \item Execute two decommitments, obtaining the opened bits:
    \begin{itemize}
        \item $(\bx', u_a) \leftarrow \extdecommit \langle \cC^*(\bx_{1,1}), \cR(\by_{c_a,0}) \rangle$ and
        \item $(\bx'',v_a) \leftarrow \extdecommit \langle \cC^*(\bx'), \cR(\by_{c_a,1}) \rangle$. 
        \item Abort if $u_a \neq v_a$.
    \end{itemize}
    \item Update state $\brho = \bx''$, and set $a = a+1$.
    \end{enumerate}
    \item If $a > i$,
    \begin{enumerate}
        \item Execute four sessions sequentially, namely:
        \begin{itemize}
            \item $\bx_{0,0}, \by_{0,0} \leftarrow \extcommit \langle \cC^*(\brho), \cR \rangle$,
            \item $\bx_{0,1}, \by_{0,1} \leftarrow \extcommit \langle \cC^*(\bx_{0,0}), \cR \rangle$, 
            \item $\bx_{1,0}, \by_{1,0} \leftarrow \extcommit \langle \cC^*(\bx_{0,1}), \cR \rangle$ and 
            \item $\bx_{1,1}, \by_{1,1} \leftarrow \extcommit \langle \cC^*(\bx_{1,0}), \cR \rangle$.
        \end{itemize}
        \item Send choice bit $c_a \leftarrow \{0,1\}$ to $\cC^*$.
        \item Execute two decommitments, obtaining the opened bits:
        \begin{itemize}
            \item $(\bx', u_a) \leftarrow \extdecommit \langle \cC^*(\bx_{1,1}), \cR(\by_{c_a,0}) \rangle$ and
            \item $(\bx'',v_a) \leftarrow \extdecommit \langle \cC^*(\bx'), \cR(\by_{c_a,1}) \rangle$. 
            \item Abort if $u_a \neq v_a$.
        \end{itemize}
        \item Update state $\brho = \bx''$, and set $a = a+1$.
    \end{enumerate}
\end{enumerate}
\paragraph{Decommit Phase.}
\begin{enumerate}
    \item Obtain $b$ from $\cC^*$, set $\brho_0 = \brho$.
    \item For $a \in [\lambda]$, obtain $\alpha_a$ from $\cC^*$. Then execute $(\brho_a, b_a) = \extdecommit \langle \cC^*_{\open}(\brho_{a-1}), \cR_{\open}(\by_{1-c_a, \alpha_a}) \rangle$, aborting if $b_a = \bot$ for any $a \in [\lambda]$. 
    \item For every $a < = i$, if 
    $d_{1-c_a,\alpha_a}^{(a)} \neq \tilde{\bot}$ and
    $b_a \neq d_{1-c_a, \alpha_a}^{(a)}$, output $\mathsf{FAIL}$.
    \item Set $b^*$ to be the majority bit in $\{b_a\}_{a \in [\lambda]}$.
\end{enumerate}
\underline{The output of $\mathsf{Hybrid}_{i,j,k}$ is $(\brho, b^*)$.}




%
\begin{claim}
There exists a negligible function $\mu(\cdot)$ such that for every $i \in [\lambda]$ and every QPT distinguisher $\cD$, 
$$|\Pr[\cD(\mathsf{Hybrid}_{i-1, 1,1}) = 1] -  \Pr[\cD(\mathsf{Hybrid}_{i, 0,1}) = 1]| = \mu(\lambda),$$
$$|\Pr[\cD(\mathsf{Hybrid}_{i, 0, 0}) = 1] -  \Pr[\cD(\mathsf{Hybrid}_{i, 0,1}) = 1]| = \mu(\lambda),$$
$$|\Pr[\cD(\mathsf{Hybrid}_{i, 0, 1}) = 1] -  \Pr[\cD(\mathsf{Hybrid}_{i, 1, 0}) = 1]| = \mu(\lambda),$$
and
$$|\Pr[\cD(\mathsf{Hybrid}_{i, 1, 0}) = 1] -  \Pr[\cD(\mathsf{Hybrid}_{i, 1, 1}) = 1]| = \mu(\lambda).$$
\end{claim}
\begin{proof} Suppose this is not the case, then there exists an adversarial committer $\cC^*$, a distinguisher $\cD$, a polynomial $p(\cdot)$, and an initial committer state $\psi$ that corresponds to a state just before the beginning of commitment $(\iota', j', k')$
where $(\iota, j, k, \iota', j', k') \in \{(i - 1, 1, 1, i, 0, 0), (i, 0, 0, i, 0, 1), (i, 0, 1, i, 1, 0), (i, 1, 0, i, 1, 1)\}$
and where $$\Pr[\cD(\mathsf{Hybrid}_{\iota, j, k}) = 1] -  \Pr[\cD(\mathsf{Hybrid}_{\iota', j',k'}) = 1]| \geq \frac{1}{p(\lambda)}.$$

Consider a reduction/adversarial committer $\widetilde{\cC}$ that obtains initial state $\psi$, then internally runs $\cC^*$, forwarding all messages between an external receiver and $\cC^*$ for the $(\iota', j', k')^{th}$ commitment. 
The commit phase then ends, with committer and receiver states resp. $(\bx_\com, \by_\com)$.
$\cR$ then begins executing the opening phase with the external receiver, and on input $\bx_\com$, continues to run the remaining commit sessions with $\cC^*$ internally -- generating for it the messages on behalf of $\cR$ according to $\hyb_{\iota, j, k}$. 
The only modification is that it forwards $\cC^*$'s opening of the $(\iota', j', k')^{th}$ commitment (if and when it is executed) to the external challenger, and records the opened bit $b'$.
After the experiment is complete including the rest of the decommit phase, $\cC^*$ outputs the final state $\bx_\final = \brho$ of the committer.

Then there is a distinguisher $\cD'$ that obtains from the challenger $(\bx_\final, b)$ (or the variable $\mathsf{FAIL}$),  runs the distinguisher $\cD$ on $(\bx_\final, b)$, and by construction, 
obtains advantage $\frac{1}{p(\lambda)}$ in distinguishing the real and ideal experiments for the $(\iota', j', k')$ commitment, concluding the proof of the claim.
\end{proof}

To complete the proof, we note that the only difference between $\mathsf{Hybrid}_{\lambda,1,1}$ and $\mathsf{Ideal}$ is the way the bit $b^*$ is computed. In more detail, in  $\mathsf{Hybrid}_{\lambda,1,1}$, the bit $b^*$ is computed as the majority of $\{e_{1-c_a,\alpha_a}^{(a)}\}_{a \in [\lambda]}$, where $e_{1-c_a,\alpha_a}^{(a)}$ is the bit extracted from commitment $e_{1-c_a,\alpha_a}$ in the $a^{th}$ iteration.
Now for every commitment strategy and every $a \in [\lambda]$, by correctness of extraction (which follows from the indistinguishability between real and ideal distributions for every commitment),
the probability that $e_{1-c_a,\alpha_a} \neq e_{c_a,\alpha_a}$ and yet the receiver does not abort in Step 2(c) in the $i^{th}$ sequential repetition, is $\leq \frac{1}{2} + \negl(\secp)$. 
Thus, this implies that the probability that $\mathsf{Hybrid}_{i,j,k}$ and $\mathsf{Ideal}$ differ is at most $2^{-\secp/2} + \negl(\secp) = \negl(\secp)$. 
Since $\mathsf{Real}$ is identical to $\mathsf{Hybrid}_{0,1,1}$, this completes the proof.

\newcommand{\eqcommit}{\mathsf{EqCommit}}
\newcommand{\eqdecommit}{\mathsf{EqDecommit}}

\section{Quantum Extractable Commitments}
\label{section: extractable com}
We construct extractable commitments by making use of the following building blocks.

\begin{itemize}
\ifsubmission
    \item We let $(\eqcommit,\eqdecommit)$ denote any statistically binding, equivocal quantum commitment scheme. Such a commitment can be obtained by applying the compiler from last section to Naor's commitment scheme~\cite{JC:Naor91}.
\else
    \item We let $(\eqcommit,\eqdecommit)$ denote any quantum statistically binding and equivocal commitment scheme, satisfying \cref{def:binding} and \cref{def:eqcom}. Such a commitment can be obtained by applying the compiler from last section to Naor's commitment scheme~\cite{JC:Naor91}.
\fi
    
    \item For a suitable polynomial $k(\cdot)$, let $h : \{0,1\}^{k(\lambda)} \times \{0,1\}^{\lambda^2} \to \{0,1\}$ be a universal hash function that is evaluated on a random seed $s \in \{0,1\}^{k(\lambda)}$ and input $x \in \{0,1\}^{\lambda^2}$.

\end{itemize}
  
Our extractable commitment scheme is described formally in Figure~\ref{fig:ext-com}. We show how to commit to a single bit, though commitment to any arbitrary length string follows by sequential repetition. Correctness of the protocol follows by inspection. In the remainder of this section, we prove the following theorem.

\ifsubmission
\begin{theorem}
\label{thm:qextcom}
\proref{fig:ext-com} describes a quantum statistically hiding and extractable bit commitment whenever $(\eqcommit,\eqdecommit)$ is instantiated with any quantum statistically binding and equivocal bit commitment scheme.
\end{theorem}
\else
\begin{theorem}
\label{thm:qextcom}
\proref{fig:ext-com} describes a quantum statistically hiding and extractable bit commitment (satisfying \cref{def:binding} and \cref{def:extcom}) whenever $(\eqcommit,\eqdecommit)$ is instantiated with any quantum statistically binding and equivocal bit commitment scheme (satisfying \cref{def:binding} and \cref{def:eqcom}).
\end{theorem}
\fi

Throughout, we will consider non-uniform adversaries, but for ease of exposition we drop the indexing by $\secp$.


\protocol
{\proref{fig:ext-com}}
{Extractable Commitment.}
{fig:ext-com}
{

\textbf{Committer $\cC$ Input:} Bit $b \in \{0,1\}$.\\

\underline{\textbf{The Protocol: Commit Phase.}}
\begin{enumerate}
\item \label[step]{step:sample} $\cC$ chooses $x \leftarrow \{0, 1\}^{2\lambda^3}$, $\theta \leftarrow \{+, \times \}^{2\lambda^3}$ and sends $\ket{x}_\theta$ to $\cR$. 


\item \label[step]{step:commit} $\cR$ chooses $\widehat{\theta} \leftarrow \{+, \times\}^{2\lambda^3}$
and obtains $\widehat{x} \in \{0, 1\}^{2\lambda^3}$ by measuring $\ket{x}_{\theta}$
in basis $\widehat{\theta}$.

$\cR$ commits to $\widehat{\theta}$ and $\widehat{x}$ position-wise: 
$\cR$ and $\cC$ execute sequentially $2\secp^3$ equivocal commitment sessions with $\cR$ as committer and $\cC$ as receiver.
That is, for each $i \in [2\secp^3]$, they compute $(\bx_{\com,i}, \by_{\com,i}) \leftarrow \eqcommit \langle \cR(\widehat{\theta}_i,\widehat{x}_i), \cC \rangle$.

\item \label[step]{step:challenge} $\cC$ sends a random test subset $T \subset [{2\lambda^3}]$ of size ${\lambda^3}$ to $\cR$. 

\item \label[step]{step:open} For every $i \in T$, $\cR$ and $\cC$ engage in 
$(\widehat{\theta}_i,\widehat{x}_i) \gets \eqdecommit\langle \cR(\bx_{\com,i}), \cC(\by_{\com,i}) \rangle$, and $\cC$ aborts if any commitment fails to open.

\item \label[step]{step:verify} $\cC$ checks that $x_i = \widehat{x}_i$ whenever $\theta_i = \widehat{\theta}_i$. If all tests pass, $\cC$ proceeds with the protocol,
otherwise, $\cC$ aborts. 

\item The tested positions are discarded by both parties: $\cC$ and $\cR$ restrict $x$ and $\theta$, respectively $\widehat{x}$ and $\widehat{\theta}$,
to the $\lambda^3$ indices $i \in \overline{T}$.
$\cC$ sends $\theta$ to $\cR$.

\item $\cC$ partitions the remaining $\lambda^3$ bits of $x$ into $\lambda$ different $\lambda^2$-bit strings $x^{(1)},\dots,x^{(\lambda)}$. For each $\ell \in [\lambda]$, sample a seed $s_\ell \gets \{0,1\}^{k(\lambda)}$ and compute $d_\ell \coloneqq h(s_\ell,x^{(\ell)})$. Then output $(s_\ell,b \oplus d_\ell)_{\ell \in [\lambda]}$.

\end{enumerate}

\underline{\textbf{The Protocol: Decommit Phase.}}
\begin{enumerate}
    \item $\cC$ sends $b$ and $(x^{(1)},\dots,x^{(\lambda)})$ to $\cR$.
    \item If either of the following fails, $\cR$ rejects and outputs $\bot$. Otherwise, $\cR$ accepts and outputs $b$.
    \begin{itemize}
        \item Let $\{s_\ell,v_\ell\}_{\ell \in [\lambda]}$ be the message received by $\cR$ in step 7. Check that for all $\ell \in [\lambda]$, $v_\ell = b \oplus h(s_\ell,x^{(\ell)})$.
        \item For each $j \in [\lambda^3]$ such that $\widehat{\theta}_j = \theta_j$, check that $\widehat{x}_j = x_j$.
    \end{itemize}
\end{enumerate}

}

\subsection{Extractability}
Consider any adversarial committer $\cC^*$ with advice $\brho$. The extractor $\cE_{\cC^*}(\brho)$ is constructed as follows. 

    \begin{enumerate}
        \item Run the first message algorithm of $\cC^*$ on input $\brho$, obtaining message $\bpsi$.
        \item For $i \in [2\lambda^3]$, 
        sequentially execute equivocal commitment sessions with the equivocal simulator $\cQ_{R^*,\com}$, where $R^*$ is the part of $\cC^*$ that participates as receiver in the $i^{th}$ session. Session $i$ results in output $(\mathbf{z}_i, \by_{\com,i})$, where $\mathbf{z}_i$ is stored by the extractor, and $\by_{\com,i}$ is the current state of $\cC^*$, which is fed as input into the next session.
        \item Obtain $T$ from $\cC^*$, and sample $\widehat{\theta} \gets \{+,\times\}^{2\secp^3}$. Let $\bpsi_i$ denote the $i^{th}$ qubit of $\bpsi$, and measure the qubits $\bpsi_i$ for $i \in T$, each in basis $\widehat{\theta}_i$. Let $\{\widehat{x}_i\}_{i \in [T]}$ be the results of the measurements.
        \item Let $\bx_\com$ be the current state of $\cC^*$. For each $i \in [T]$, execute $\cQ_{R^*,\open}((\widehat{\theta}_i,\widehat{x}_i),\mathbf{z}_i, \bx_\com)$, where $R^*$ is the part of $\cC^*$ that participates in the $i^{th}$ opening, and $\bx_\com$ is updated to be the current state of $\cC^*$ after each sequential session.
        \item If $\cC^*$ aborts at any point, abort and output $\bot$, otherwise continue.
        \item Discard tested positions and restrict $\widehat{\theta}$ to the indices in $\overline{T}$. Obtain $\theta \in \{+,\times\}^{\lambda^3}$ from $\cC^*$.
        Measure the qubits $\bpsi_i$ in basis $\theta_i$ to obtain $\widehat{x}_i$ for $i \in \overline{T}$, and then partition $\widehat{x}$ into $\lambda$ different $\lambda^2$-bit strings $\widehat{y}_1, \ldots, \widehat{y}_\lambda$.
        \item Obtain $\{s_\ell, v_\ell\}_{\ell \in [\secp]}$ from $\cC^*$. 
        Let $b^*$ be the most frequently occurring bit in $\{ h(s_\ell, \widehat{x}^{(\ell)}) \oplus v_\ell\}_{\ell \in [\secp]}$. Output $(\bx_\com,\by_\com,b^*)$, where $\bx_\com$ is the resulting state of $\cC^*$ and $\by_\com = (\theta, \widehat{\theta},\widehat{x})$.
    \end{enumerate}
    \ifsubmission
    We now prove that $\cE_{\cC^*}$ is a secure extractor; for space reasons, a full definition of extractability in the quantum setting is in the full version. 
    \else
    We now prove that $\cE_{\cC^*}$ is a secure extractor in the sense of~\cref{def:extcom}. We do so via a sequence of hybrids.
    \fi
    \\

        \noindent {\textbf{$\mathsf{Hyb}_1$.}} Define distribution $\mathsf{Hyb}_1$ identically to $\mathsf{Real}$ (the honest interaction), except that in Step 2, for $i \in [2\lambda^3]$, sequentially execute equivocal commitment sessions using the equivocal simulator $\cQ_{R^*,\com}$, as described in the extractor. In Step 4, for every $i \in T$, open the $i^{th}$ commitment to $(\widehat{\theta}_i,\widehat{x}_i)$ using $\cQ_{R^*,\open}$, as described in the extractor.
        
    
        
        Suppose there exists a polynomial $p(\cdot)$ such that
        $$\Big|\Pr[\cD^*(\bsigma,\mathsf{Hyb}_1) = 1] - \Pr[\cD^*(\bsigma,\mathsf{Hyb}_0) = 1] \Big|  \geq \frac{1}{p(\secp)}$$
        then consider a sequence of sub-hybrids $\hyb_{0,0} = \hyb_0$ and $\hyb_{0,1}, \hyb_{0,2}, \ldots, \hyb_{0,2\lambda^3}$, where for every $j \in [1,2\lambda^3]$, $\hyb_{0,j}$ is identical to $\hyb_0$ except for $i \in [j]$, we sequentially execute equivocal commitment sessions using the equivocal simulator $\cQ_{R^*,\com}$, as described in the extractor. For $i \in [j+1, 2\lambda^3]$, the commitments are constructed identically to $\hyb_{0}$. 
        In Step 4, for every $i \in T$ such that $i \leq j$, open the $i^{th}$ commitment to $(\widehat{\theta}_i,\widehat{x}_i)$ using $\cQ_{R^*,\open}$, as described in the extractor. Open remaining commitments $i \in T$ identically to $\hyb_0$.
        Then, we have that for large enough $\lambda$, there exists an index $i = i(\lambda) \in [2\lambda^3]$ such that: $$\Big|\Pr[\cD^*(\bsigma,\mathsf{Hyb}_{0,i-1}) = 1] - \Pr[\cD^*(\bsigma,\mathsf{Hyb}_{0,i}) = 1] \Big|  \geq \frac{1}{p(\secp) \cdot 2\lambda^3}$$
        This implies that there is an initial state $\psi$, which is the state of the receiver just prior to the $i^{th}$ equivocal commitment, and a reduction $\cR$ initialized with $\psi$ that internally executes $\cR^*$, externally forwarding only the commit phase of the $i^{th}$ commitment. 
        It then ends the commit phase with the external committer. Next, it begins the open phase, executing all other commitments in Step 2 as part of the open phase internally with $\cR^*$. It also sends a randomly sampled $T$ in Step 3, and executes the decommitments in Step 4 internally with $\cR^*$, except forwarding the opening of the $i^{th}$ commitment externally. It completes the rest of the hybrid according to $\hyb_{0,i-1}$ and outputs its final state as $\bx_{\final}$. Then the existence of $\cD^*$ that distinguishes this state $\bx_\final$ between the cases when the $i^{th}$ commitment was generated honestly vs. equivocated, as described above, contradicts the equivocality of this commitment. 
        Therefore, we conclude that there is a negligible function $\mu(\cdot)$ such that:
        $$\Big|\Pr[\cD^*(\bsigma,\mathsf{Hyb}_1) = 1] - \Pr[\cD^*(\bsigma,\mathsf{Hyb}_0) = 1] \Big|  = \mu(\secp)$$

        \noindent {\textbf{$\mathsf{Hyb}_2$.}} This is identical to $\mathsf{Hyb}_1$,  except that the verifier measures qubits of $\ket{x}_{\theta}$ {\em only after} obtaining a description of the set $T$, and \emph{only measures} the qubits $i \in [T]$.
        The output of this experiment is identical to $\mathsf{Hyb}_1$, therefore for any QPT distinguisher $(\cD^*,\bsigma)$, $$\Pr[\cD^*(\bsigma,\mathsf{Hyb}_3) = 1] = \Pr[\cD^*(\bsigma,\mathsf{Hyb}_2) = 1].$$
        
        Moreover, the only difference between $\mathsf{Hyb}_2$ and $\mathsf{Ideal}$ is that $\mathsf{Ideal}$ outputs $\mathsf{FAIL}$ when the 
        message $b$ opened by $\cC^*$ is not $\bot$ and differs from the one extracted by $\cE_{\cC^*}$ .
        Therefore, to derive a contradiction it will suffice to prove that there exists a negligible function $\nu(\cdot)$ such that
        $$\Pr[\mathsf{FAIL}|\mathsf{Ideal}] = \nu(\lambda).$$
        
        Consider any sender $\cC^*$ that produces a committer state $\bx_\com$ and then decommits to message $b'$ using strings $(y_1,\dots,y_\secp)$ during the decommit phase.
        Let $T' \subseteq [\lambda]$ denote the set of all indices $\ell \in [\lambda]$ such that the corresponding $x^{(\ell)} \neq v_\ell$, where $\widehat{x}^{(\ell)}$ denotes the values obtained by the extractor in Step 6. Then we have the following claim.
        \begin{claim}
        There exists a negligible function $\nu(\cdot)$ such that $$\Pr[|T'| > \lambda/2 ] = \nu(\lambda)$$
        where the probability is over the randomness of the extractor.
        \end{claim}
        \begin{proof}
        For every $\ell \in [\lambda]$, we have that (over the randomness of the extractor):
        $$\Pr \Big[\cR_{\mathsf{open}}(\by_{\mathsf{com}}) \text{ outputs }\bot\text{ in }\langle \cC^*_{\mathsf{open}} (\bx_{\mathsf{com}}), \cR_{\mathsf{open}} (\by_{\mathsf{com}}) \rangle \ \Big| \ x^{(\ell)} \neq \widehat{x}^{(\ell)} \Big] \geq \frac{1}{2}.$$ Indeed, the receiver will reject if for some position $i$ for which $x^{(\ell)} \neq \widehat{x}^{(\ell)}$, it holds that $\theta_i = \widehat{\theta}_i$. Since $\widehat{\theta}$ was sampled uniformly at random, this will occur for a single $i$ with independent probability 1/2. This implies that $\Pr[|T'| > \lambda/2] \leq \frac{1}{2^{\lambda/2}}$, and the claim follows.
        \end{proof}
        By construction of $\cE_{\cC^*}$, $\Pr[\mathsf{FAIL}|\mathsf{Ideal}] < \Pr[|T'| > \lambda/2]$, and therefore it follows that there exists a negligble function $\nu(\cdot)$ such that $$\Pr[\mathsf{FAIL}|\mathsf{Ideal}] = \nu(\lambda).$$

\ifsubmission
The proof of that the extractable commitment scheme described in~\cref{fig:ext-com} is statistically hiding follows readily from quantum sampling techniques developed by~\cite{C:BouFeh10}, and is deferred to the full version.
\else

\subsection{Statistical Hiding}
\label{subsec:stat-hiding}

In this subsection, we prove that our extractable commitment scheme described in~\cref{fig:ext-com} is statistically hiding, assuming that $\eqcommit$ is a statistically binding (equivocal) commitment scheme.

Our proof closely follows techniques of~\cite{C:BouFeh10}, who developed techniques for analyzing~\cite{FOCS:CreKil88}-style oblivious transfer protocols based on quantum sampling techniques. We import a lemma from~\cite{C:BouFeh10}.

\paragraph{\cite{C:BouFeh10} Quantum Sampling.} 

Consider the following interaction between a player $P$ and a challenger $C$.
\begin{enumerate}
    \item $P$ prepares an arbitrary quantum state $\ket{\psi}_{\gray{X}\gray{Y}}$ where $\gray{X}$ is a $2n$-qubit register and $\gray{Y}$ is an arbitrary-size register of $P$'s choice. 
    
    $P$ sends the contents of $\gray{X}$ to $C$ along with two length-$2n$ classical strings $\hat{x} \in \{0,1\}^{2n}$ and $\hat{\theta} \in \{+,\times\}^{2n}$.
    
    \item\label[step]{step:check-T} $C$ samples a uniformly random set of $n$ indices $T \subset [2n]$ as well as a uniformly random $\theta \in \{+,\times\}^{2n}$.
    
    For each $i \in T$ it measures the $i$th qubit of $\gray{X}$ in basis $\hat{\theta}_i$. If there exists $i \in T$ where $\theta_i = \hat{\theta}_i$ and the measurement outcome does not match $\hat{x}_i$, then $C$ aborts the experiment and outputs $\bot$. Otherwise, $C$ sends $\overline{T} \coloneqq [2n] \setminus T$ to $P$.
    
    \item $P$ responds with a subset $S \subseteq \overline{T}$. 
    
    \item\label[step]{step:measure} $C$ measures the indices of $\gray{X}$ specified by $S$, obtaining a length-$|S|$ bitstring $X_S$ of measurement outcomes; let $s$ denote the number of indices $i \in S$ where $\theta_i \neq \hat{\theta}_i$. 
\end{enumerate}

\begin{lemma}[\cite{C:BouFeh10} Quantum Sampling Lemma]
\label{lemma:sampling}
For any computationally unbounded player $P$ and $0 < \varepsilon < 1$, the probability the sampling experiment does not abort and
\[ \mathbf{H}_{\infty}(X_S \mid \gray{X}_{\bar{S}}, \gray{Y}) > s- \varepsilon \cdot |S|. \]
is at most $\negl(|S|^2/n)$.
\end{lemma}

\begin{proof}(Sketch.)
This lemma follows readily from techniques developed in~\cite{C:BouFeh10} to relate quantum sampling to classical sampling. Following~\cite{C:BouFeh10}, it suffices to analyze the case where $\hat{x} = 0^{2n}$ and $\hat{\theta} = +^{2n}$, since all other settings of $\hat{x}$ and $\hat{\theta}$ follow from an identical argument under a change of basis. 

In \cref{step:check-T} of the quantum sampling experiment, $C$ verifies that for every index in a random choice of $n$ indices $T \subset [2n]$, measuring the corresponding qubit in the computational basis yields $0$. As observed in~\cite{C:BouFeh10}, this is a natural quantum analogue of the following classical sampling experiment (\cite[Example 4]{C:BouFeh10}) on a length-$2n$ bitstring $X$:
\begin{enumerate}
    \item randomly select a size-$n$ subset $T \subset [2n]$,
    \item randomly select a string $\theta \in \{0,1\}^{2n}$,
    \item verify that $X_i = 0$ on each index $i \in T$ where $\theta_i = 1$, and abort otherwise.
\end{enumerate}
Let $X_{\overline{T}}$ denote the string $X$ restricted to the $n$ indices $\overline{T} = [2n] \setminus T$. For any initial bitstring $X$, the ``classical error probability'' is the probability that the sampling experiment succeeds and $X_{\overline{T}}$ has relative Hamming weight $\geq \delta$ is at most $6e^{-n \delta^2/50}$ (bound derived in \cite[Appendix B.4]{C:BouFeh10}).

\cite[Theorem 3]{C:BouFeh10} demonstrates that the ``quantum error probability'' is at most the square root of the classical error probability. That is, if the verification in \cref{step:check-T} passes, then the residual state on $\gray{X}$ restricted to $\overline{T}$ is within trace distance $\sqrt{6e^{-n \delta^2/50}}$ of a superposition of states with relative Hamming weight at most $\delta$.

Suppose for a moment that the state were \emph{exactly} a superposition of states with relative Hamming weight at most $\delta$. Now consider the restriction of $\gray{X}$ to a subset of indices $S \subset \overline{T}$ where exactly $s$ positions have $\hat{\theta}_i \neq \theta_i$. The relative Hamming weight of the state when restricted to $S$ is at most $\delta n/ |S|$. As long as $\delta n/ |S| < 1/2$, by~\cite[Corollary 1]{C:BouFeh10}, we have
\[ \mathbf{H}_{\infty}(X_S \mid \gray{X}_{\overline{S}}, \gray{Y}) > s - h(\delta n/|S|)\cdot |S|\]
where $h: [0,1] \rightarrow [0,1]$ is the binary entropy function defined as $h(p) = p \log (p) + (1-p) \log (1-p)$ for $0 < p < 1$ and as $0$ for $p = 0$ or $1$. For any desired $\varepsilon \in (0,1)$, there exists an $\epsilon' \in (0,1/2)$ such that $h(\varepsilon') = \varepsilon$, so we can plug in $\delta = \varepsilon' |S|/n$ to achieve the stated min-entropy lower bound of $s - \varepsilon \cdot |S|$.
\end{proof}


\paragraph{Proof of Statistical Hiding.} We now prove statistical hiding of our extractable commitment scheme in~\cref{fig:ext-com}. The core proof idea is due to~\cite{C:BouFeh10}.

Consider a \emph{purification} of the committer's strategy from~\cref{fig:ext-com}, where instead of sampling random $x \gets \{0,1\}^{2\secp^3}, \theta \gets \{+,\times\}^{2\secp^3}$ and sending $\ket{x}_\theta$ to $\mathcal{R}$ in~\cref{step:sample}, the committer $\mathcal{C}$ generates $2\secp^3$ EPR pairs $(\ket{00} + \ket{11})/\sqrt{2}$, and sends $\mathcal{R}$ one half of each pair. This change has no effect on the receiver's view. In this modification of~\cref{step:sample}, the committer will still sample random $\theta \gets \{+,\times\}^{2\secp^3}$ but does not use it to perform any measurements (and hence $x$ is not yet determined). We now recast~\cref{step:sample,step:open,step:commit,step:challenge,step:verify} of~\cref{fig:ext-com}, i.e., the ``measurement-check subprotocol'' to correspond with this undetectable modification of the committer's behavior:


\begin{enumerate}
\item $\cC$ chooses $2\lambda^3$ EPR pairs $(\ket{00} + \ket{11})/\sqrt{2}$ and sends $\mathcal{R}$ one half of each pair. The committer samples random $\theta \gets \{+,\times\}^{2\lambda^3}$ but does not perform any measurements (and hence $x$ is not yet determined).
\item $\cR$ chooses $\widehat{\theta} \leftarrow \{+, \times\}^{2\lambda^3}$
and obtains $\widehat{x} \in \{0, 1\}^{2\lambda^3}$ by measuring its $2\lambda^3$ EPR pair halves in bases specified by $\widehat{\theta}$.

$\cR$ commits to $\widehat{\theta}$ and $\widehat{x}$ using \emph{statistically binding commitments}.

\item $\cC$ sends a random test subset $T \subset [{2\lambda^3}]$ of size ${\lambda^3}$ to $\cR$. 

\item For every $i \in T$, $\cR$ decommits to $(\widehat{\theta}_i,\widehat{x}_i)$, and $\cC$ aborts if any decommitment is invalid.

\item $\cC$ checks that $x_i = \widehat{x}_i$ whenever $\theta_i = \widehat{\theta}_i$. If all tests pass, $\cC$ proceeds with the protocol,
otherwise, $\cC$ aborts. 
\end{enumerate}

Notice that $\cC$'s messages are completely independent of its bit $b$ up to this point. If $\cR$ passes the check above, then $\cC$ will obtain a $\secp^3$-bit string $x$ by measuring its qubits (i.e., EPR pair halves) in positions $\overline{T} \coloneqq [2\secp] \setminus T$ in the bases specified in $\theta$. $\cC$ then partitions $x$ into $\secp$ different $\secp^2$-bit strings strings $x^{(1)},\dots,x^{(\secp)}$ and sends the receiver
\[ (s_1,b \oplus h(s_1,x^{(1)})),\dots,(s_\secp,b \oplus h(s_\secp,x^{(\secp)})), \]
where each $s_1,\dots,s_\secp$ is a random seed for the universal hash function $h$. By the leftover hash lemma (\cref{thm:privacy-amplification}), $b$ will be statistically hidden if each $x^{(\ell)}$ has sufficient independent min-entropy.

To that end, observe (following~\cite[Section 5]{C:BouFeh10}) that the committer decides to accept the receiver's commitment openings $(\hat{\theta}_i,\hat{x}_i)_{i \in [T]}$ \emph{only if} a certain ``quantum sampling experiment'' succeeds. Since the commitments are statistically binding, the values of $\hat{x},\hat{\theta}$ that $\mathcal{C}$ would accept are statistically determined after $\mathcal{R}$ commits. Let $\ket{\psi}_{\gray{X}\gray{Y}}$ be the joint state of the committer and receiver where $\gray{X}$ corresponds to the committer's $2\secp^3$ registers, and $\gray{Y}$ corresponds to the receiver's state. Then the committer's check can only pass if the sampling strategy of~\cref{lemma:sampling} succeeds for the committed values of $\hat{x},\hat{\theta}$. We remark that while no efficient $\mathcal{C}$ can actually \emph{perform} this sampling strategy (as it involves breaking the computational hiding of the underlying equivocal commitments), its behavior is statistically indistinguishable from $\mathcal{C}$ that does perform this sampling. 

By applying~\cref{lemma:sampling}, we conclude that each $x^{(\ell)}$ must therefore have min-entropy $\Omega(\lambda^2)$ conditioned on $\mathcal{R}$'s view and on all $x^{(j)}$ for $j \neq \ell$, since by Hoeffding's inequality, $\theta$ and $\hat{\theta}$ differ in at least $\Omega(\lambda^2)$ places within the $\lambda^2$ indices corresponding to $x^{(\ell)}$ with $1-\negl(\lambda)$ probability. Applying~\cref{thm:privacy-amplification} (leftover hash lemma) one by one to each $x^{(\ell)}$, this implies that each $(s_\ell,h(s_\ell,x^{(\ell)}))$ is statistically indistinguishable from uniform even given the receiver's view, and thus $b$ is statistically hidden.


\fi

\section{Quantum Oblivious Transfer from Extractable and Equivocal Commitments}
\label{sec: OT from eecom}
\renewcommand{\commit}{\mathsf{EECom}}

\subsection{Definitions for Oblivious Transfer with Quantum Communication}
\label{sec:otdef}
An oblivious transfer with quantum communication is a protocol between a quantum interactive sender $\cS$ and a quantum interactive receiver $\cR$, where the sender $\cS$ has input $m_0,m_1 \in \{0,1\}^\secp$ and the receiver $\cR$ has input $b \in \{0,1\}$. After interaction the sender outputs $(m_0,m_1)$ and the receiver outputs $(b,m_b)$. 

Let $\cF(\cdot,\cdot)$ be the following functionality. $\cF(b,\cdot)$ takes as input either $(m_0,m_1)$ or $\abort$ from the sender, returns $\mathsf{end}$ to the sender, and outputs $m_b$ to the receiver in the non-$\abort$ case and $\bot$ in the $\abort$ case. $\cF(\cdot,(m_0,m_1))$ takes as input either $b$ or $\abort$ from the receiver,
returns $m_b$ to the receiver, and
returns $\mathsf{end}$ to the sender  
in the non-$\abort$ case, and returns $\bot$ to the sender  in the $\abort$ case. 

\begin{definition}
\label{def:ot}
We let $\langle S(m_0, m_1), R(b) \rangle$ denote an execution of the OT protocol with sender input $(m_0, m_1)$ and receiver input bit $b$. 
We denote by $\brho_{\out,S^*}\langle S^*(\brho), R (b) \rangle$ and $\mathsf{OUT}_{R}\langle S^*(\brho), R(b) \rangle$ the final state of a non-uniform malicious sender $S^*(\brho)$ and the output of the receiver $R(b)$ at the end of an interaction (leaving the indexing by $\secp$ implicit). We denote by $\brho_{\out,R^*}\langle S(m_0, m_1), R^*(\brho) \rangle$ and $\mathsf{OUT}_{S}\langle  S(m_0, m_1), R^*(\brho) \rangle$ the final state of a non-uniform malicious receiver $R^*(\brho)$ and the output of the sender $S(m_0,m_1)$ at the end of an interaction.
We require OT to satisfy the following security properties:
\begin{itemize}
\item \noindent {\bf Receiver Security.} 
For every QPT non-uniform malicious sender $S^*$, there exists a simulator $\simu_{S^*}$ such that the following holds. For any non-uniform advice $\brho,\bsigma$ where $\brho$ and $\bsigma$ may be entangled, bit $b \in \{0,1\}$, and QPT non-uniform distinguisher $D^*$, $\simu_{S^*}(\brho)$ sends inputs $(m_0, m_1)$ or $\mathsf{abort}$ to the ideal functionality ${\fot(b,\cdot)}$, and outputs a final state $\brho_{\mathsf{Sim},\out,S^*}$. The output of the ideal functionality to the receiver in this experiment is denoted by $\mathsf{OUT}_R$. It must hold that
\begin{align*}
\bigg|&\Pr\left[D^*\left(\bsigma,\left(\brho_{\mathsf{Sim},\out,S^*}, \mathsf{OUT}_{R} \right)\right) = 1\right]\\
&-\Pr\left[D^*\left(\bsigma,\left(\brho_{\out,S^*}\langle S^*(\brho), R(b) \rangle, \mathsf{OUT}_{R}\langle S^*(\brho), R(b) \rangle\right) \right) = 1\right]\bigg| = \negl(\secp).
\end{align*}


\item \noindent {\bf Sender Security.} For every QPT non-uniform malicious receiver $R^*$, there exists a  simulator $\simu_{R^*}$ such that the following holds. For any non-uniform advice $\brho,\bsigma$ where $\brho$ and $\bsigma$ may be entangled, pair of sender inputs $(m_0, m_1)$, and QPT non-uniform distinguisher $D^*$, $\simu_{R^*}(\brho)$ sends bit $b$ or $\mathsf{abort}$ to the ideal functionality ${\fot(m_0,m_1,\cdot)}$, and outputs a final state $\brho_{\mathsf{Sim},\out,R^*}$. The output of the ideal functionality to the sender in this experiment is denoted by $\mathsf{OUT}_S$. It must hold that 
\begin{align*}
\bigg|&\Pr\left[D^*\left(\bsigma,\left(\brho_{\mathsf{Sim},\out,R^*}, \mathsf{OUT}_S \right)\right) = 1\right]\\
&-\Pr\left[D^*\left(\bsigma,\left(\brho_{\out,R^*}\langle S(m_0,m_1), R^*(\brho) \rangle, \mathsf{OUT}_S\langle S(m_0,m_1), R^*(\brho) \rangle\right) \right) = 1\right]\bigg| = \negl(\secp).
\end{align*}

\end{itemize}
\end{definition}

\subsection{Our Construction}

We construct simulation-secure quantum oblivious transfer by making use of the following building blocks.

\begin{itemize}
\ifsubmission
    \item Let $(\eecommit,\eedecommit)$ denote any quantum bit commitment scheme satisfying extractability and equivocality. Such a commitment scheme may be obtained by applying the compiler from~\cref{sec: equivocal commitment} to the extractable commitment constructed in~\cref{section: extractable com}. 
\else
    \item Let $(\eecommit,\eedecommit)$ denote any quantum bit commitment scheme satisfying extractability (\cref{def:extcom}) and equivocality (\cref{def:eqcom}). Such a commitment scheme may be obtained by applying the compiler from~\cref{sec: equivocal commitment} to the extractable commitment constructed in~\cref{section: extractable com}. 
\fi
    \item Let $h: \{0,1\}^{k(\secp)} \times \mathcal{X} \rightarrow \{0,1\}^\secp$ be a universal hash with seed length $k(\secp) = \poly(\secp)$ and domain $\mathcal{X}$ the set of all binary strings of length \emph{at most} $8 \secp$. 
\end{itemize}

Our QOT protocol is described in Protocol \ref{fig:ot}, which is essentially the~\cite{FOCS:CreKil88} protocol instantiated with our extractable and equivocal commitment scheme.


\protocol
{\proref{fig:ot}}
{Quantum Oblivious Transfer.}
{fig:ot}
{
\textbf{Sender $S$ Input:} Messages $m_0, m_1 \in \{0,1\}^{\lambda} \times \{0,1\}^{\lambda}$\\
\textbf{Receiver $R$ Input:} Bit $b \in \{0,1\}$\\

\underline{\textbf{The Protocol:}}
\begin{enumerate}
\item $S$ chooses $x \leftarrow \{0, 1\}^{16\lambda}$ and $\theta \leftarrow \{+, \times \}^{16\lambda}$ 
and sends $\ket{x}_\theta$ to $R$. 

\item $R$ chooses $\widehat{\theta} \leftarrow \{+, \times\}^{16\lambda}$
and obtains $\widehat{x} \in \{0, 1\}^{16\lambda}$ by measuring $\ket{x}_{\theta}$
in basis $\widehat{\theta}$. Then, $S$ and $R$ execute $16\lambda$ sessions of $\eecommit$ sequentially with $R$ acting as committer and $S$ as receiver. In session $i$, $R$ commits to the bits $\widehat{\theta}_i,\widehat{x}_i$.


\item $S$ sends a random test subset $T \subset [16\lambda]$ of size $8\lambda$ to $R$. 

\item For each $i \in T$, $R$ and $S$ sequentially execute the $i$'th $\eedecommit$, after which $S$ receives the opened bits $\widehat{\theta}_i,\widehat{x}_i$.


\item $S$ checks that $x_i = \widehat{x}_i$ whenever $\theta_i = \widehat{\theta}_i$. If all tests pass, $S$ accepts,
otherwise, $S$ rejects and aborts.

\item The tested positions are discarded by both parties: $S$ and $R$ restrict $x$ and $\theta$, respectively $\widehat{x}$ and $\widehat{\theta}$, to the $8\secp$ indices $i \in \overline{T}$.
$S$ sends $\theta$ to $R$.

\item $R$ partitions the positions of $\overline{T}$ into two parts: the ``good'' subset $I_b = \{i : \theta_i = \widehat{\theta}_i\}$ and the ``bad'' subset $I_{1-b} = \{i : \theta_i \neq \widehat{\theta}_i\}$. $R$ sends $(I_0,I_1)$ to $S$.

\item $S$ samples seeds $s_0,s_1 \leftarrow \{0,1\}^{k(\lambda)}$ and sends $\left(s_0,h(s_0,x_0) \oplus m_0,s_1,h(s_1,x_1) \oplus m_1\right)$, where $x_0$ is $x$ restricted to the set of indices $I_0$ and $x_1$ is $x$ restricted to the set of indices $I_1$.

\item $R$ decrypts $s_b$ using $\widehat{x}_b$, the string $\widehat{x}$ restricted to the set of indices $I_b$.
\end{enumerate}
}

\ifsubmission
\begin{theorem}
\label{theorem:QOT}
The protocol in Figure \ref{fig:ot} is a simulation-secure QOT protocol whenever $(\eecommit,\eedecommit)$ is instantiated with a quantum bit commitment satisfying extractability and equivocality.
\end{theorem}
\else
\begin{theorem}
\label{theorem:QOT}
The protocol in Figure \ref{fig:ot} is a QOT protocol satisfying Definition \ref{def:ot} whenever $(\eecommit,\eedecommit)$ is instantiated with a quantum bit commitment satisfying extractability (\cref{def:extcom}) and equivocality (\cref{def:eqcom}).
\end{theorem}
\fi

We prove that the resulting QOT protocol satisfies standard simulation-based notions of receiver and sender security. 

\ifsubmission
The proof of sender security follows readily from quantum sampling techniques developed by~\cite{C:BouFeh10}, and is deferred to the full version.
\else
\fi

\subsection{Receiver Security}
    Consider any adversarial sender $S^*$ with advice $\brho$. The simulator $\mathsf{Sim}_{S^*}(\brho)$ is constructed as follows.
    \begin{enumerate}
        \item Run the first message algorithm of $S^*$ on input $\brho$ to obtain message $\bpsi$. 
        \item Execute $16\lambda$ sequential sessions of $\eecommit$. In each session, run the equivocator $\cQ_{\cR^*,\com}$, where $\cR^*$ denotes the portion of $S^*$ that participates as receiver in the $i^{th}$ sequential $\eecommit$ session.
        \item Obtain test subset $T$ of size $8\lambda$ from $S^*$.
        \item For each $i \in T$, sample $\widehat{\theta}_i \leftarrow \{+, \times\}$. 
        Obtain $\widehat{x}_i$ by measuring the $i^{th}$ qubit of $\bpsi$ in basis $\widehat{\theta}_i$.
        For each $i \in T$, sequentially execute the equivocal simulator $\cQ_{\cR^*, \open}$ on input $(\widehat{\theta}_i,\widehat{x}_i)$ and the state obtained from $\cQ_{\cR^*, \com}$.
        \item If $S^*$ continues, discard positions indexed by $T$. Obtain $\theta_i$ for $i \in \overline{T}$ from $S^*$, and compute $x_i$ for $i \in \overline{T}$ by measuring the $i^{th}$ qubit of $\bpsi$ in basis $\theta_i$.
        \item For every $i \in \overline{T}$, sample bit $d_i \leftarrow \{0,1\}$. Partition the set $\overline{T}$ into two subsets $(I_0, I_1)$, where for every $i \in \overline{T}$, place $i \in I_0$ if $d= 0$ and otherwise place $i \in I_1$. Send $(I_0, I_1)$ to $S$.
        \item Obtain $(y_0, y_1)$ from $S$. Set $x_0$ to be $x$ restricted to the set of indices $I_0$ and $x_1$ to be $x$ restricted to the set of indices $I_1$.
        For $b \in \{0,1\}$, parse $y_b = (s_b, t_b)$ and compute $m_b = t_b \oplus h(s_b, x_b)$. 
        \item If $S^*$ aborts anywhere in the process, send $\mathsf{abort}$ to the ideal functionality. Otherwise, send $(m_0, m_1)$ to the ideal functionality. Output the final state of $S^*$.
    \end{enumerate}
%

    
     Next, we establish receiver security according to \cref{def:ot}. Towards a contradiction, suppose there exists a 
    bit $b \in \zo$,
    a non-uniform QPT sender $(S^*,\brho)$, a non-uniform QPT distinguisher $(D^*,\bsigma)$, and polynomial $\mathsf{poly}(\cdot)$ s.t.
    \begin{align*}\Big|&\Pr\left[D^*\left(\bsigma,\left(\brho_{\mathsf{Sim},\out,S^*}, \mathsf{OUT}_{R}\right)\right) = 1\right] \\ &- \Pr\left[D^*\left(\bsigma,\left(\brho_{\out,S^*}\langle S^*(\brho), R(b) \rangle, \mathsf{OUT}_{R}\langle S^*(\brho), R(b) \rangle\right)\right) = 1\right] \Big| \geq \frac{1}{\mathsf{poly}(\secp)}.\end{align*}
    Fix any such $b$, sender $(S^*,\brho)$ and distinguisher $(D^*,\bsigma)$. We derive a contradiction via an intermediate hybrid experiment, defined as follows with respect to bit $b$ and sender $(S^*,\brho)$.\\
    
    \noindent {\textbf{$\mathsf{Hyb}$.}} In this hybrid, we generate the QOT receiver commitments via the equivocal simulator $\cQ_{\cR^*}$ (where $\cR^*$ is derived from the malicious QOT sender $S^*$), and otherwise follow the honest QOT receiver's algorithm. 
    
    \begin{enumerate}
        \item Run the first message algorithm of $S^*$ on input $\brho$ to obtain message $\bpsi$. 
        \item Choose $\widehat{\theta} \leftarrow \{+, \times\}^{16\lambda}$ and obtain $\widehat{x} \in \{0,1\}^{16\lambda}$ by measuring $\bpsi$ in basis $\widehat{\theta}$.
        Execute $16\lambda$ sequential sessions of $\eecommit$. In each session, run the equivocator $\cQ_{\cR^*,\com}$, where $\cR^*$ denotes the portion of $S^*$ that participates as receiver in the $i^{th}$ sequential $\eecommit$ session.
        \item Obtain test subset $T$ of size $8\lambda$ from $S^*$.
        \item For each $i \in T$, sequentially execute the equivocal simulator $\cQ_{\cR^*, \open}$ on input $\widehat{\theta}_i,\widehat{x}_i$ and the state obtained from $\cQ_{\cR^*, \com}$.
        \item If $S^*$ continues, discard positions indexed by $T$. Obtain $\theta_i$ for $i \in \overline{T}$ from $S^*$.
        \item Partition the set $\overline{T}$ into two subsets: the ``good'' subset $I_b = \{i:\theta_i = \widehat{\theta}_i\}$ and the ``bad'' subset $I_{1-b} = \{i: \theta_i \neq \widehat{\theta}_i\}$. Send $(I_0, I_1)$ to $S$.
        \item Obtain $(y_0, y_1)$ from $S$. Set $x_b$ to be $\widehat{x}$ restricted to the set of indices $I_b$, and compute and set $m_b = t_b \oplus h(s_b, x_b)$. If $S^*$ aborts anywhere in the process, let $\bot$ be the output of the receiver, otherwise let $m_b$ be the output of the receiver. 
    \end{enumerate}
    The output of $\mathsf{Hyb}$ is the joint distribution of the final state of $S^*$ and the output of the receiver.
Receiver security then follows from the following two claims.

\begin{claim}
    $\Pr\left[D^*\left(\bsigma,\left(\brho_{\mathsf{Sim},\out,S^*}, \mathsf{OUT}_{R}\right)\right) = 1\right] \equiv \Pr\left[D^*(\bsigma,\mathsf{Hyb}) = 1\right].$
\end{claim}
\begin{proof} The only differences in the simulated distribution are (1) that measurements of $S^*$'s initial message $\bpsi$ are delayed (which cannot be noticed by $S^*$), and (2) a syntactic difference in that the ideal functionality is queried to produce the receiver's output. 
\end{proof}
    
\begin{claim}
    There exists a negligible function $\nu(\cdot)$ such that
    $$\Big|\Pr[D^*(\bsigma,\mathsf{Hyb}) = 1] - \Pr\left[D^*\left(\bsigma,\left(\brho_{\out,S^*}\langle S^*(\brho), R(b) \rangle, \mathsf{OUT}_{R}\langle S^*(\brho), R(b) \rangle\right)\right) = 1\right] \Big| = \nu(\secp).$$
\end{claim}

\begin{proof}
    The only difference between the two distributions is that in the first, the receiver generates commitments according to the honest commit algorithms of $\eecommit$ while in the second, commitments in step 2 are generated via the equivocal simulator $\cQ_{\cR^*}$ of $\eecommit$. Therefore, this claim follows by the equivocality of 
    $(\eecommit,\eedecommit)$ .
\end{proof}
    
\ifsubmission
\else
\subsection{Sender Security} 
    Consider any adversarial receiver $R^*$ with quantum auxiliary state $\brho$. The simulator $\simu_{R^*}(\brho)$ is constructed as follows:
    \begin{enumerate}
        \item Choose $x \leftarrow \{0,1\}^{16\lambda}$ and $\theta \leftarrow \{+,\times\}^{16\lambda}$. Send $\ket{x}_{\theta}$ to $R^*$. 
        \item Execute $16\lambda$ sequential sessions of $\eecommit$. 
        In the $i^{th}$  session for $i \in [16\lambda]$, run the extractor $\cE_{\cC^*_i, \com}$ where $\cC^*_i$ denotes the portion of $R^*$ that participates as committer in the $i^{th}$ sequential $\eecommit$ session.
        Obtain from $\cE_{\cC^*_i, \com}$ the extracted values $(\widehat{\theta}_i, \widehat{x}_i)$.
        \item Sample and send a random test subset $T \subset [16\lambda]$ of size $8\lambda$ to $R^*$.
        \item For each index in $T$ execute the opening phases $\eedecommit$ with $R^*$, and abort if any fail. Discard the opened values, and relabel the remaining indices from $1$ to $8 \secp$.
        
        \item Check that $x_i = \widehat{x}_i$ whenever $\theta_i = \widehat{\theta}_i$, and if not abort.
        \item Send $\theta$ to $R^*$.
        \item Obtain $R^*$'s message, which specifies sets $I_0$ and $I_1 = [8\secp] \setminus I_0$. Let $S$ be the set of indices in $I_0$ such that $\theta_i \neq \widehat{\theta}_i$. If $|S| \geq 3\secp/2$ then set $b = 1$ and otherwise set $b=0$; intuitively, if $|S| \geq 3\secp/2$, we treat $I_0$ as the ``bad'' subset $I_{1-b}$ where $\theta_i \neq \hat{\theta}_i$ for many indices $i$.
        \item Obtain the output $m_b$ of $\fot$ on input bit $b$. Set $m_{1-b} = 0^{\lambda}$.
        \item Let $x_0,x_1$ denote $x$ restricted to indices $I_0,I_1$ respectively. Sample seeds $s_0,s_1 \leftarrow \{0,1\}^{k(\lambda)}$ and send $s_0,h(s_0,x_0) \oplus m_0,s_1,h(s_1,x_1) \oplus m_1$. Output the final state of $R^*$.
    \end{enumerate}
    

We will now establish sender security according to \cref{def:ot}. Towards a contradiction, suppose there exists a  pair of messages $(m_0, m_1) \in \{0,1\}^{\secp}$, a non-uniform QPT receiver $(R^*,\brho)$, a non-uniform QPT distinguisher $(D^*,\bsigma)$, and a polynomial $\mathsf{poly}(\cdot)$ such that 
    \begin{align*}\Big|&\Pr\left[D^*\left(\bsigma,\left(\brho_{\mathsf{Sim},\out,R^*}, \mathsf{OUT}_S\right)\right) = 1\right] \\&- \Pr\left[D^*\left(\bsigma,\left(\brho_{\out,R^*}\langle S(m_0, m_1), R^*(\brho) \rangle, \mathsf{OUT}_{S} \langle S(m_0, m_1), R^*(\brho) \rangle\right)\right) = 1\right] \Big| \geq \frac{1}{\mathsf{poly}(\secp)}.
    \end{align*}
    Fix any such $(m_0, m_1)$, receiver $(R^*,\brho)$ and distinguisher $(D^*,\bsigma)$. We derive a contradiction via an intermediate hybrid experiment, defined as follows. Like the simulator, this hybrid simulates and extracts from the receiver's commitments, but does not yet replace the sender string $m_{1-b}$ with $0^\lambda$. \\
    



\noindent $\mathsf{Hyb}$. Consider any adversarial receiver $R^*$ with advice $\brho$. The output of $\mathsf{Hyb}$ depends on $(m_0, m_1)$ and is generated as follows.
\begin{enumerate}
        \item Choose $x \leftarrow \{0,1\}^{16\lambda}$ and $\theta \leftarrow \{+,\times\}^{16\lambda}$. Send $\ket{x}_{\theta}$ to $R^*$. 
        \item Execute $16\lambda$ sequential sessions of $\eecommit$. 
        In the $i^{th}$  session for $i \in [16\lambda]$, run the extractor $\cE_{\cC^*_i, \com}$ where $\cC^*_i$ denotes the portion of $R^*$ that participates as committer in the $i^{th}$ sequential $\eecommit$ session.
        Obtain from $\cE_{\cC^*_i, \com}$ the extracted values $\widehat{\theta}_i, \widehat{x}_i$.
        \item Sample and send a random test subset $T \subset [16\lambda]$ of size $8\lambda$ to $R^*$.
        \item For each index in $T$ execute the opening phases $\eedecommit$ with $R^*$, and abort if any fail. Discard the opened values, and relabel the remaining indices from $1$ to $8 \secp$.
        \item Check that $x_i = \widehat{x}_i$ whenever $\theta_i = \widehat{\theta}_i$, and if not abort.
        \item Obtain $R^*$'s message, which specifies sets $I_0$ and $I_1$. 
        \item Let $x_0,x_1$ denote $x$ restricted to indices $I_0,I_1$ respectively. Sample seeds $s_0,s_1 \leftarrow \{0,1\}^{k(\lambda)}$ and send $s_0,h(s_0,x_0) \oplus m_0,s_1,h(s_1,x_1) \oplus m_1$. Output the final state of $R^*$.
    \end{enumerate}

This hybrid distribution is indistinguishable to the real distribution due to the extractability (Definition \ref{def:extcom}) of the commitment. Indeed, the values extracted from the simulator are used in place of the values opened by $R^*$. Conditioned on the opening phases succeeding, the joint distribution of the view of $R^*$ and these values will be equal with all but negligible probability. This can be proved via a sequential hybrid argument similar to the one in the proof of Theorem \ref{thm:equiv-compiler}.


Now, the only difference between this hybrid experiment and the simulation is the value of the string $m_{1-b}$, where $b$ is the bit extracted in the simulated game. Thus, to complete the proof, it suffices to argue that in the simulation, the string $h(s_{1-b},x_{1-b})$, which is used to mask $m_{1-b}$, is statistically close to a uniformly random string. 

The remaining argument is essentially identical in structure to the argument of statistical hiding of the extractable commitment in~\cref{subsec:stat-hiding}, which in turn is a direct application of quantum sampling techniques developed by~\cite{C:BouFeh10}. We therefore sketch the following statistical steps, highlighting where the argument differs from~\cref{subsec:stat-hiding}.

The idea, as in~\cref{subsec:stat-hiding}, is to consider a purification of the strategy for the challenger implementing $\mathsf{Hyb}$. In particular, suppose instead of the challenger generating $\ket{x}_{\theta}$ by sampling $x,\theta$ at random, it generates $16\secp$ EPR pairs and sends $R^*$ half of each pair, keeping the other half of each pair for itself. The challenger still samples $\theta \gets \{+,\times\}^{16\secp}$, but does not yet measure to determine $x$. When the challenger implementing $\mathsf{Hyb}$ needs to check $x_i = \hat{x}_i$, it obtains $x_i$ by measuring its $i$th EPR pair half in basis $\hat{\theta}_i$. Instead of invoking statistical binding to extract the committed values of $\hat{x},\hat{\theta}$ as done in~\cref{subsec:stat-hiding}, we can simply run the $\mathsf{EECommit}$ extractor to extract $\hat{x},\hat{\theta}$. Now, if $R^*$ can successfully open its commitments and the challenger implementing $\mathsf{Hyb}$ does not abort, this corresponds to a successful execution of the quantum sampling experiment in~\cref{lemma:sampling} for $n = 16\secp$ and the extracted values of $\hat{x},\hat{\theta}$.

Now, consider the two subsets $I_0,I_1 \subset [8\secp]$ sent by $R^*$. For $c \in \{0,1\}$, let $\mathsf{wt}(I_c)$ be the number of indices in $I_c$ such that $\theta_i \neq \widehat{\theta}_i$. If $\mathsf{wt}(I_0) > 3\secp/2$, then $b$ will be set to 1, and by~\cref{lemma:sampling} for $S = I_0$ and $s = \mathsf{wt}(I_0) > 3\secp/2$ followed by an application of the leftover hash lemma (\cref{thm:privacy-amplification}), $(s_0,h(s_0,x_0))$ will be statistically close to a uniformly random string. Otherwise, we have $\mathsf{wt}(I_0) \leq 3\secp/2$. By Hoeffding's inequality, with probability $1 - \negl(\secp)$, the strings $\theta$ and $\widehat{\theta}$ must differ in more than $3\lambda$ positions, which means in this case we have $\mathsf{wt}(I_1) > 3\secp/2$. As a result, we have $b = 0$, and by the same argument as above, $(s_1,h(s_1,x_1))$ is statistically close to a uniformly random string.


\fi

Finally, Theorems \ref{thm:equiv-compiler}, Theorem \ref{thm:qextcom}, and Theorem \ref{theorem:QOT} give the following.
\begin{corollary}
\label{cor:two}
Quantum oblivious transfer (QOT) satisfying Definition \ref{def:ot} can be based on black-box use of statistically binding bit commitments, or on black-box use of quantum-hard one-way functions.
\end{corollary}

\section{Acknowledgments}
The authors are grateful to the Simons Institute programs on {\em Lattices: Algorithms, Complexity and Cryptography}, and {\em The Quantum Wave in Computing} for fostering this collaboration. Thanks also to Alex Grilo, Huijia Lin, Fang Song, and Vinod Vaikuntanathan for discussions about similarities and differences with \cite{GLSV}. A.C. is supported by Vannevar Bush Faculty Fellowship N00014-17-1-3025. This material is based on work supported in part by DARPA under Contract No. HR001120C0024 (for DK). Any opinions, findings and conclusions or recommendations expressed in this material are those of the author(s) and do not necessarily reflect the views of the United States Government or DARPA.

\ifsubmission
\bibliographystyle{splncs04}
\else
\bibliographystyle{alpha}
\fi
\bibliography{abbrev3,custom,crypto,main}

\appendix

\newpage

\clearpage

\end{document}